\documentclass[11pt]{article}

\usepackage[margin=1in]{geometry}
\usepackage{amsmath,amssymb}
\usepackage{mathtools}
\usepackage{bbm}
\usepackage{graphicx}
\usepackage{natbib}
\usepackage{enumitem}

\usepackage{amsthm}
\usepackage{aliascnt}
\usepackage[backref=page]{hyperref}
\usepackage[nameinlink,noabbrev,capitalise]{cleveref}

\usepackage{physics}
\usepackage{braket}
\usepackage[margin=1in]{geometry}
\usepackage{amsmath,amssymb,amsfonts,setspace,url}
\usepackage{bm}
\usepackage{comment}
\usepackage{algpseudocode}
\usepackage{latexsym}
\usepackage{subfig}
\usepackage{float}
\usepackage[only,shortleftarrow,shortrightarrow]{stmaryrd}
\usepackage{fancybox}
\usepackage{bigstrut,array,multirow}
\usepackage{here}
\usepackage{color}
\usepackage{url}
\usepackage{mathpazo}
\usepackage{extpfeil}
\usepackage{xcolor}
\usepackage{tikz}
\usetikzlibrary{arrows.meta}
\usepackage[disable]{todonotes}
\usepackage[bb=boondox]{mathalfa}

\newtheorem{theorem}{Theorem}[section]

\newaliascnt{corollary}{theorem}
\newtheorem{corollary}[corollary]{Corollary}
\aliascntresetthe{corollary}

\newaliascnt{proposition}{theorem}

\aliascntresetthe{proposition}

\newaliascnt{definition}{theorem}
\newtheorem{definition}[definition]{Definition}
\aliascntresetthe{definition}

\newaliascnt{fact}{theorem}
\newtheorem{fact}[fact]{Fact}
\aliascntresetthe{fact}

\newaliascnt{conjecture}{theorem}

\aliascntresetthe{conjecture}

\newaliascnt{remark}{theorem}

\aliascntresetthe{remark}

\newtheorem{problem}{Problem}

\newaliascnt{claim}{theorem}
\newtheorem{claim}[claim]{Claim}
\aliascntresetthe{claim}

\newaliascnt{observation}{theorem}
\newtheorem{observation}[observation]{Observation}
\aliascntresetthe{observation}

\newaliascnt{lemma}{theorem}
\newtheorem{lemma}[lemma]{Lemma}
\aliascntresetthe{lemma}

\crefname{theorem}{Theorem}{Theorems}
\crefname{corollary}{Corollary}{Corollaries}
\crefname{proposition}{Proposition}{Propositions}
\crefname{definition}{Definition}{Definitions}
\crefname{property}{Property}{Properties}
\crefname{fact}{Fact}{Facts}
\crefname{conjecture}{Conjecture}{Conjectures}
\crefname{remark}{Remark}{Remarks}
\crefname{problem}{Problem}{Problems}
\crefname{claim}{Claim}{Claims}
\crefname{open}{Open Problem}{Open Problems}
\crefname{observation}{Observation}{Observations}
\crefname{lemma}{Lemma}{Lemmas}

\Crefname{theorem}{Theorem}{Theorems}
\Crefname{corollary}{Corollary}{Corollaries}
\Crefname{proposition}{Proposition}{Propositions}
\Crefname{definition}{Definition}{Definitions}
\Crefname{property}{Property}{Properties}
\Crefname{fact}{Fact}{Facts}
\Crefname{conjecture}{Conjecture}{Conjectures}
\Crefname{remark}{Remark}{Remarks}
\Crefname{problem}{Problem}{Problems}
\Crefname{claim}{Claim}{Claims}
\Crefname{open}{Open Problem}{Open Problems}
\Crefname{observation}{Observation}{Observations}
\Crefname{lemma}{Lemma}{Lemmas}

\newcommand{\mingyang}[1]{{\color{orange} \footnotesize(Mingyang: #1)}}
\newcommand{\ananta}[1]{{\color{purple} \footnotesize(Ananta: #1)}}

\newcommand{\brakett}[2]{\left\langle#1\middle|#2\right\rangle}
\newcommand{\E}{\mathbb E}
\newcommand{\Prb}{\mathrm{Pr}}
\newcommand{\MBEH}{\mathsf{BiMBEH}}
\newcommand{\BMC}{\mathsf{BiMBCert}}
\newcommand{\Alg}{\mathsf{Alg}}
\newcommand{\rec}{\mathrm{rec}}
\newcommand{\inp}{\mathrm{inp}}
\newcommand{\loc}{\mathrm{loc}}
\newcommand{\blk}{\mathrm{blk}}
\newcommand{\wt}{\operatorname{wt}}
\newcommand{\Ran}{\operatorname{Ran}}
\newcommand{\spanop}{\operatorname{span}}
\newcommand{\ot}{\otimes}
\newcommand{\proj}[1]{|#1\rangle\!\langle #1|}
\newcommand{\dist}{\text{dist}}

\newcommand{\Id}{\mathbb{I}}

\newcommand{\perpp}{{\perp}}

\newcommand{\CalO}{\mathcal{O}}
\newcommand{\CalT}{\mathcal{T}}
\newcommand{\CalR}{\mathcal{R}}
\newcommand{\CalP}{\mathcal{P}}
\newcommand{\CalH}{\mathcal{H}}

\newcommand{\CalD}{\mathcal{D}}

\newcommand{\trianglelisting}{\textsl{TriangleListing}}
\newcommand{\spanner}{\textsl{SpannerConstruction}}
\newcommand{\rmif}{\textsl{RMIF}}
\newcommand{\bhvf}{\textsl{BHVF}}
\newcommand{\ind}{\mathbf{1}}
\newcommand{\Bin}{\operatorname{Bin}}
\title{Quantum Query Lower Bounds for Triangle-Listing and Spanners}
\author{
Yu Chen\thanks{National University of Singapore.
Email: \texttt{yu.chen@nus.edu.sg}.}
\qquad
Ananta Mukherjee\thanks{Centre for Quantum Technologies, National University of Singapore.
Email: \texttt{a.mukherjee@u.nus.edu}.}
\qquad
Mingyang Yang\thanks{National University of Singapore.
Email: \texttt{myangat@u.nus.edu}.}
}
\date{}

\begin{document}
\maketitle

\begin{abstract}
This paper gives quantum query lower bounds for two relational graph problems, triangle listing and explicit multiplicative spanner construction, in the general graph query model, where quantum adjacency, degree and neighborhood queries are all available in arbitrary superposition. Both results are obtained by reductions via intermediate multi-block search problems. We extend the quantum query recording framework by Zhandry (CRYPTO 2019) and Hamoudi and Magniez (ToCT 2023) to handle a recording architecture for bidirectional oracles and give a generic blockwise soundness framework, which for arbitrary families of local accepting projectors, gives an exact operator-norm characterization of their maximum overlap with the subspace of bounded weight records. Using the intermediate search problems, we exhibit a family of $n$-vertex graphs with $\Theta(n)$ triangles on which listing any constant fraction of the triangles requires $\Omega(n^{3/2-o(1)})$ quantum queries. This is the first nontrivial quantum lower bound for triangle listing, a question raised by Jiang and Peng (ICML 2026). Further, we show that, for every fixed $k\ge 7$, constructing a multiplicative $k$-spanner requires $\Omega(n^{1+\frac{1}{2\mu_k}})$ quantum queries, where $\mu_k=k/3+O(1)$. For $k\in\{7,8\}$ the bound is $\Omega(n^{5/4})$, which matches what would be implied by an unproven instance of the Erd\H{o}s girth conjecture, and for large $k$ the exponent $1+\frac{3}{2k}$ exceeds the $1+\frac{4}{3k}$ implied by the provable high girth
dense graphs due to Lazebnik, Ustimenko and Woldar, 1995.
\end{abstract}
\newpage
\tableofcontents
\newpage
\section{Introduction}
For a problem, where the input is accessible via quantum oracle queries, 
quantum query complexity asks what is the optimum number of oracle accesses required to the input before a valid answer can be produced. For graph problems, the answer depends strongly on how the graph is accessed. In the standard setting, the graph is stored in a black box that answers the following types of queries.
\begin{itemize}
    \item Adjacency queries: given two vertices $u,v \in V$, return whether there exists an edge $(u,v)\in E$.
    \item Degree queries: given a vertex $u \in V$, return the degree of $u$.
    \item Neighborhood queries: given a vertex $u \in V$ and an index $i$, return the $i$-th neighbor of $u$ (the ordering of the neighbors is arbitrary, but fixed).
\end{itemize}
Each interface exposes information of a different shape, and a lower bound proved against one of them need not survive when the others are also available. We call the model in which all three oracles can be used simultaneously the \emph{mixed model}, also known as the general graph query model. In the quantum version of the model, the oracles can be queried in superposition to combine global tests of chosen vertex pairs with local exploration of chosen neighborhoods, and a lower bound in this model must rule out every combination of those strategies.

The quantum query complexity of graph problems has a rich history. The seminal work of D\"urr, Heiligman, H\o yer and Mhalla \cite{DHHM06} determined the quantum query complexity of connectivity, strong connectivity, minimum spanning tree and single-source shortest paths, in the adjacency model and in the neighborhood model treated separately. Since then, quantum speedups have been developed for a long list of graph search tasks, including triangle finding \cite{findtriangle, sotafindtriangle}, maximum matching \cite{kimmel2021query, blikstad2022nearly}, and graph sparsification and Laplacian solving \cite{apers2022quantum}. 
\\
But, non trivial lower bounds results, especially in the mixed graph query model have been few and far between. The two standard techniques, the polynomial method \cite{beals2001quantum} and the adversary method \cite{Ambainis00}, are majorly convenient for decision problems. But, relational search and construction
problems, especially in the graph setting, present a different difficulty.  Their outputs are combinatorial
objects, there may be many valid answers, and different successful executions
may reveal different parts of the input.  A lower-bound argument must therefore
control an entire family of possible outputs rather than one
fixed accepting event.
\\The recording-query method offers a natural way to reason about such problems.
After a change of representation, information exposed by a quantum query is
written into an auxiliary record, while the algorithm itself remains fully
quantum and unrestricted.  The method was introduced by Zhandry~\cite{Zhandry19}and was later
developed as a lower-bound tool for multi-solution query problems by Hamoudi
and Magniez~\cite{HamoudiMagniez2023}.  A recording proof separates two questions.  How
quickly can a bounded number of queries create records containing information
about many input blocks?  And how much recorded information is necessary for a
successful output?
\\~\\
In this paper, first we extend the recording method to handle a bidirectional query where two coupled queries can access a shared block record. We then prove a generic soundness lemma which, for an arbitrary family of local accepting projectors, gives an exact characterization of their maximum overlap with the subspace of records of bounded weight. This results in recording-based lower bounds for multi-block search problems that help us to obtain  reductions from search and construction flavored graph problems to those intermediate search problems while preserving the full mixed interface: adjacency, degree and neighborhood queries available in arbitrary superposition. 
\\
Using the above architecture, 
in this paper we prove quantum query lower bounds for two fundamental relational graph problems. In the \textit{triangle listing} problem, studied in ~\cite{bjorklund2014listing, jiang2026quantum}, the task is to report the  triangles present in an $n$ vertex graph. No non trivial quantum query lower bound was known for this problem~\cite{jiang2026quantum}. We showed a quantum query lower bound of $\Omega(n^{3/2-o(1)})$ for an algorithm that lists at least a constant fraction of the triangles when the number of triangles are linear in n. In the \textit{multiplicative $k$-spanner construction} problem, the algorithm must output actual edges of the unknown graph forming a spanning subgraph that stretches every pairwise distance by a factor of at most $k$. By provable instances of the girth conjecture of Erd\H{o}s~\cite{erdos1964extremal}, trivial lower bounds were known for $k\leq 6$ and $k=9,10$. For generic $k$, provable instances of high girth dense graphs by \cite{LUW95} implies a lower bound of $\Omega(n^{1+4/3k+O(1/k^2)})$. In this paper, 
for generic $k$, we showed a quantum query lower bound of $\Omega(n^{1+3/2k+O(1/k^2)})$ and a lower bound of $\Omega(n^{5/4})$ for 7 and 8 stretch spanners which are particularly interesting as they match the conjectured lower bound by an unproven instance of the girth conjecture.
\\
We describe below the main results along with the context laid by relevant prior works.
\subsection{Our Results and Related Works}
\paragraph{Multi-Block Search with Bidirectional Oracles.} Both graph applications rest on quantum query lower bounds for two search problems over block-structured inputs, defined in \Cref{sec:two-problems}. The main search problem underlying our recording development is the Bidirectional Hidden-Value Finding problem ($(B,N,d)$-$\bhvf$, \Cref{subsec:BHVF}), which is used for the $k-$spanner reduction. Here, each block hides a uniformly random pair $(a,b)\in[N]\times[N]$, and the block can be probed from either side: probing position $x$ in the forward direction reveals $b$ if $x=a$ and nothing otherwise, and symmetrically for the backward direction. The task is to recover the full hidden pair for a $d$ fraction of the blocks. The triangle-listing reduction
uses the simpler one-sided \emph{Restricted Marked-Item Finding} problem
$(B,N,d)$-$\rmif$ where the input consists of $B$ blocks. Each block independently hides one uniformly random marked position among $N$, accessible through a membership oracle. The task is to return the marked positions of a $d$ fraction of the blocks. 

\begin{theorem}[Informal version of \Cref{thm:bhvf}, \Cref{thm:rmif}]\label{thm:informal-search}
For every constant $d\in(0,1]$, solving 
$(B,N,d)$-$\bhvf$ or $(B,N,d)$-$\rmif$ with success probability at least $2/3$ requires $\Omega(B\sqrt{N})$ quantum queries.
\end{theorem}




\paragraph{Two Contributions to the Recording Method.} \Cref{thm:informal-search} is proved with the recording method, introduced by Zhandry \cite{Zhandry19} and developed into a general tool for search lower bounds by Hamoudi and Magniez \cite{HamoudiMagniez2023}. Recording techniques have also appeared in 
\cite{belovs2026tight, Gilani2026Quantum, carolan-compressed} for combinatorial and graph problems.
Our proof contributes two ingredients that we believe are of independent interest.
\begin{enumerate}[label=(\arabic*)]
     \item \emph{A recording architecture for bidirectional oracles.} The two query directions of a $BHVF$ block are not two unrelated search oracles. They describe the same hidden pair, and a successful probe in either direction reveals the value hidden in the other one. We build a block-level recording representation in which both directions read and write one shared record per block, and the change of representation remains exact after any number of queries. The algorithm stays unrestricted: it may choose the block, the direction and the probed position in superposition, with arbitrary computation between queries.
    \item \emph{A generic soundness lemma.} Every recording-based lower bound must at some point show that success forces a heavy record. We isolate this step in \Cref{lem:technical lemma 1} and prove it as an exact operator identity, not an inequality. The lemma takes an arbitrary test in each block, given by a projector of arbitrary rank, and computes the largest overlap between the joint accepting subspace and any subspace of low record weight. The answer is a classical lower-tail probability for independent per-block coins. Because it is an operator-norm statement, it automatically accounts for arbitrary superpositions of, and entanglement across, the blocks. The lemma separates the problem-specific work of identifying the accepting subspace from the generic claim that a light record cannot pass many tests, and it can be applied off the shelf in other recording lower bounds.
\end{enumerate}
The details are discussed in \Cref{sec:BHVF-lb-proof} and \Cref{sec:appnd:RMIF-lb}.
\paragraph{Triangle Listing.}
Triangle listing is a well studied subgraph problem, studied classically since the 1970s \cite{itai1978finding, chiba1985arboricity, alon1997finding}. The fastest known classical listing algorithms are due to Bj\"orklund, Pagh, Vassilevska Williams and Zwick \cite{bjorklund2014listing}, and their running times are conditionally optimal under the 3SUM conjecture \cite{patrascu2010towards}. Fine-grained connections to all-pairs shortest paths appeared in \cite{williams2020monochromatic}. On the quantum side, work has concentrated on the decision version, triangle \emph{finding}, which has served for two decades as a benchmark for quantum query algorithms. Search-based algorithms \cite{buhrman2001quantumfDurr} were followed by quantum walks \cite{szegedy2004quantum, magniez2007search, findtriangle}, span programs and learning graphs \cite{belovs2012constant, jeffery2013nested, lee2013improved}, leading to Le Gall's $\tilde{O}(n^{5/4})$ bound \cite{sotafindtriangle}, with refinements for sparse graphs \cite{legall2017sparse} and in logarithmic factors \cite{carette2020extended}. Hamoudi and Magniez~\cite{HM19Chebyshev} studied the related problem of triangle counting and gave an optimal
quantum algorithm using $\widetilde{O}\!\left(\sqrt{n}/t^{1/6}+m^{3/4}/\sqrt{t}\right)$ queries, via their quantum Chebyshev inequality.
Their task is estimation of the number of triangles rather than recovery of the triangles themselves. Very recently, Jiang and Peng \cite{jiang2026quantum} initiated the study of quantum algorithms for triangle listing. They list all $t$ triangles of a graph with $n$ vertices and $m$ edges in time $\tilde{O}\big(\min(n^{5/4}t^{7/12}+n^{7/6}t^{7/9},\; m+m^{3/4}t^{1/2},\; n^{3/2}t^{1/2})\big)$, and they use listing as a primitive for triangle cut sparsification and triangle-based clustering. 
\\On the lower bound side, results have been few and far between. For triangle finding, the best known quantum bound is $\Omega(n)$ \cite{belovs2014power}. For triangle listing, no nontrivial quantum lower bound was known at all, and \cite{jiang2026quantum} raise the question explicitly. Our first application is a progress towards answering this question.

\begin{theorem}[Informal version of \Cref{thm:triangle listing main}]\label{thm:informal-triangle}
There is a family of $n$-vertex graphs, each containing $\Theta(n)$ triangles, on which listing any constant fraction of the triangles with success probability at least $2/3$ requires
$\Omega\left(\frac{n^{3/2}}{2^{O(\sqrt{\log n})}}\right)=\Omega\left(n^{3/2-o(1)}\right)$
quantum queries in the mixed model.
\end{theorem}
The technical details are given in \Cref{sec:triangle-listing-reduction}.
On graphs of the relevant density ($n^{2-o(1)}$ edges with $t=\Theta(n)$ triangles), the best known upper bound is $\tilde{O}(n^{5/4}t^{7/12}+n^{7/6} t^{7/9})=\tilde{O}(n^{35/18})$ \cite{jiang2026quantum}; closing the gap between $\tilde{O}(n^{35/18})$ and $\Omega(n^{3/2-o(1)})$ remains an interesting open problem.

\paragraph{Multiplicative $k$-Spanner Construction.} Given an integer $k\geq 1$, a multiplicative $k$-spanner of a graph $G$ is a spanning subgraph $H$ with $\dist_H(u,v)\le k\cdot \dist_G(u,v)$ for every pair of vertices $u,v$. The construction task is explicit: the algorithm must output actual edges of the unknown graph. Spanners were introduced by Peleg and Sch\"affer \cite{peleg1989graph}, having appeared implicitly in the study of network synchronizers \cite{awerbuch1985complexity, peleg1989optimal}. Every graph admits a $(2k-1)$-spanner with $O(n^{1+1/k})$ edges via the greedy algorithm \cite{althofer1993sparse}, such spanners can be computed in near-linear time \cite{baswana2007simple}, and they underlie approximate distance oracles \cite{thorup2005approximate} and many further applications surveyed in \cite{ahmed2020graph}. Quantum algorithms enter this picture through the work of Apers and de Wolf \cite{apers2022quantum}: their quantum algorithms for cut approximation and Laplacian solving are powered by a quantum spanner-construction subroutine running in time $\tilde{O}(\sqrt{mn})$, which is sublinear in the number of edges $m$ for dense graphs. This raises a basic question: how many quantum queries are \emph{necessary} to construct a $k$-spanner?

To the best of our knowledge, before this work, the only known route to a lower bound went through the length of the output. If every cycle of $G$ is strictly longer than $k+1$, then no $k$-spanner of $G$ can omit any edge (\Cref{lem:mandatory-edge-criterion}). A dense graph of girth at least $k+2$ therefore forces a long answer, and the question becomes extremal: how many edges can a graph of high girth have? The girth conjecture of Erd\H{o}s \cite{erdos1964extremal} asserts that there are $n$-vertex graphs of girth $2\kappa+2$ with $\Omega(n^{1+1/\kappa})$ edges, which would be tight \cite{bondy1974cycles}. The conjecture is proven exactly for $\kappa\in\{1,2,3,5\}$ \cite{erdos1966problem, brown1966graphs, benson1966minimal, wenger1991extremal} and remains unproven everywhere else including $\kappa=4$. Translated to the stretch parameter, the proven cases are $k\in\{1,2\}:\Omega(n^2)$, $k\in\{3,4\}=\Omega(n^{3/2})$, $k\in\{5,6\}:\Omega(n^{4/3})$ and $k\in\{9,10\}:\Omega(n^{6/5})$. The pair $k\in\{7,8\}$ is the first missing case, and the conjectured $k$-spanner lower bound of $\Omega(n^{1+2/k})$ or $\Omega(n^{1+2/(k+1)})$ for even and odd k respectively is open for every stretch beyond 10. Unconditionally, for generic $k$, the densest known high-girth graphs are those of Lazebnik, Ustimenko and Woldar \cite{LUW95} (see also \cite{furedi2013history}). Through the mandatory-edge route, they give a lower bound of $\Omega(n^{6/5})$ for $k\in\{7,8\}$, and a lower bound of $\Omega(n^{1+\frac{4}{3k}+O(1/k^2)})$ for large $k$. Our second application takes a different route and proves a stronger bound.
\begin{theorem}[Informal version of \Cref{thm:generic-lower-bound}]\label{thm:informal-spanner}
For every fixed integer $k\ge 7$, constructing an explicit multiplicative $k$-spanner with success probability at least $2/3$ requires
$\Omega\left(n^{1+\frac{1}{2\mu_k}}\right)$
quantum queries in the mixed model, where
$\mu_k=\left\lfloor\frac{k+1}{3}\right\rfloor-\ind[k\equiv 2 \pmod 6]=\frac{k}{3}+O(1)$. In particular, the bound is $\Omega(n^{5/4})$ for $k\in\{7,8\}$, and $\Omega\left(n^{1+\frac{3}{2k}+O(1/k^2)}\right)$ for large $k$.
\end{theorem}
The technical details are given in \Cref{sec:k-spanner-reduction}.
Three remarks are due on \Cref{thm:informal-spanner}. First, for $k\in\{7,8\}$ the theorem gives, unconditionally, the $n^{5/4}$ scale that the open girth 10 case of the Erd\H{o}s conjecture would only predict through output length. It thus fills the first gap left by the unproven instances of the girth conjecture, and it does so without constructing any high-girth graph. Second, for large $k$ the exponent $1+\frac{3}{2k}$ is polynomially larger than the exponent $1+\frac{4}{3k}$ available from the graphs of \cite{LUW95}. Third, the hard instances behind the theorem have girth four, and each of them, in fact, admits a $3$-spanner with $O(n)$ edges. The bound is therefore not driven by a long answer, and the proof neither relies on nor produces dense graphs of high girth. The obstruction is local: a constant fraction of the hidden edges have no replacement path of length at most $k$, even though the graph is full of short cycles and always has a sparse valid spanner. What the approach measures is genuinely the query cost of \emph{locating} a valid spanner.

\subsection{Technical Overviews}
\subsubsection{Quantum Query Lower Bounds}
\paragraph{The setup.}
The technical heart of the quantum lower bounds is \Cref{sec:BHVF-lb-proof}, which proves the query lower bound for the Bidirectional Hidden-Value Finding (BHVF) problem defined in \Cref{subsec:BHVF}. An instance consists of $B$ independent blocks. Each block $t\in [B]$ hides a pair of values $(r^A_t,r^B_t)$ sampled uniformly from the support $[N]\times [N]$. Block $t$ is accessed through the two functions
\(
A^r_t(x)=(r^B_t+1)\cdot \ind[x=r^A_t],
\:\:
B^r_t(x)=(r^A_t+1)\cdot \ind[x=r^B_t],
\)
where the additional shift by 1 is considered for technical reasons.
 The task is to output the full hidden pair for at least a $d$ fraction of the blocks, for a given constant $d$. A quantum algorithm may query all blocks inside one superposed query, build up correlations across blocks, and defer every decision to a single final measurement. The companion bound for the simpler RMIF problem (\Cref{thm:rmif}) follows from the same template. Essentially, we get \Cref{thm:informal-search}.


\paragraph{The recording framework.}
Classically, a recording argument keeps a log of the positions the algorithm has probed, and argues that a short log cannot pin down many blocks. Quantumly, queries are made in superposition and no such log exists. Following the compressed-oracle framework of Zhandry \cite{Zhandry19} and its extension by Hamoudi and Magniez \cite{HamoudiMagniez2023}, we instead purify the input: the algorithm runs against a uniform superposition of all instances, and the input register itself becomes the log. 
The first two subsections of \Cref{sec:BHVF-lb-proof} set this up for our problem. Essentially, the joint initial state is $\ket{0}_{\Alg}\otimes(\ket{\tilde u}\,\ket{\tilde u})^{\otimes B}$ with $\ket{\tilde u}=\frac1{\sqrt N}\sum_{i\in[N]}\ket i$,
Every input coordinate is granted one extra symbol $\perp$, read as ``nothing recorded yet''. We call this enlarged space the recording space. The change of viewpoint is a single fixed local unitary: $S$ which exchanges $\ket\perpp\leftrightarrow\ket{\tilde u}$ and fixes the orthogonal complement of their span. Conjugating the query oracle by this unitary produces the recording oracle. 
\\
First we note that, both query directions of a block act on the same pair of recorded coordinates, so the record of a block collects everything the algorithm has learned about it from either side. Second, and crucially, nothing is lost in this change of basis: \Cref{thm:relation-phi-psi} argues that after any number of queries the true joint state equals the recorded joint state up to this fixed unitary. Every statement about the real execution can therefore be proved on the recorded one, whose basis states carry a transparent meaning: a block is \emph{vacuum} when both of its recorded coordinates are $\perp$ (i.e., the block is $\ket{\perp\perp}$), and \emph{nonvacuum} otherwise. We call the number of nonvacuum blocks the \emph{record weight} of a basis state. It is our formal proxy for how many blocks the algorithm has genuinely touched. The lower bound then rests on two pillars: the record weight grows slowly, and a light record caps the success probability.

\paragraph{Bounding the recording progress.}
The first pillar states that few queries produce a light record. We measure progress after $t$ queries by the amplitude that the recorded state places on basis states of record weight at least $k$. Two observations drive the analysis. A query acts on a single block, so one query raises the weight by at most one. Moreover, raising it is costly: an overlap computation shows that a single query moves amplitude at most $2/\sqrt{N}$ out of the vacuum of the queried block. Together they imply

\begin{lemma}[Progress bound; informal version of \Cref{lem:progress-recurrence',lem:binomial-progress'}]
After $t$ queries, the amplitude on basis states of record weight at least $k$ is at most $\binom{t}{k}\left(2/\sqrt{N}\right)^{k}$.
\end{lemma}

In particular, if the number of queries is at most a small constant times $B\sqrt{N}$, then the amplitude on record weight at least $dB/2$ is $2^{-\Omega(B)}$. In words: after few queries, almost all of the recorded state knows only a small fraction of the blocks.

\paragraph{Soundness of recording.}
The second pillar states that a light record is fatal for a query algorithm that claims high success probability. Any recording-based lower bound must eventually argue that a high record weight is \emph{necessary} for success; otherwise a slowly growing record proves nothing. This soundness step is often the delicate part of such arguments, because the success measurement and the record-weight decomposition do not commute. Our soundness lemma settles it in an exact and fully generic form.

\begin{lemma}[Soundness of recording; informal version of \Cref{lem:technical lemma 1}]
Fix a subset $D$ of blocks and, for each block $t\in D$, an arbitrary test given by a projector $Q_t$ on the record of that block. Let $p_t$ be the measure of the overlap of the test $Q_t$ with the vacuum. Then the largest probability with which a state of record weight at most $k$ passes all the tests simultaneously is exactly
$\Pr\left[\sum_{t\in D} X_t\le k\right]$,
where the $X_t\in\{0,1\}$ are independent Bernoulli random variable with $\Pr[X_t=0]=p_t$.
\end{lemma}

Intuitively, a block whose record is vacuum passes its test only through the vacuum overlap $p_t$, and a block escapes this penalty only by being one of the at most $k$ nonvacuum blocks. The lemma says these are the only two options, and that the trade-off between them is governed by independent per-block coins. Three features are important about this lemma. It is an identity rather than an inequality: the relevant operator norm is computed exactly, through an explicit orthogonal decomposition into joint eigenspaces. It is generic: the statement never mentions our specific problem, and the per-block tests $Q_t$ are arbitrary, so the lemma serves as an off-the-shelf soundness step in any recording lower bound where high record weight must be shown necessary. Finally, it converts an operator question about two non-commuting projectors into a purely probabilistic one, so standard concentration bounds apply immediately.

For our problem BHVF, the test at block $t$ is that the algorithm's guess for this block is correct. A guess that pins down one specific pair has vacuum overlap exactly $1/N^{2}$, so a still-vacuum block is guessed correctly with probability $1/N^{2}$ and no better. Since a successful output must name at least $dB$ blocks, a Chernoff bound turns the soundness lemma into the following statement: any state of record weight at most $dB/2$ passes the success measurement with probability at most $e^{-\Omega(B)}$ (\Cref{lem:main-mbeh helper'}).

\paragraph{From progress and soundness to the lower bound.}
The two pillars combine mechanically in the last subsection of \Cref{sec:BHVF-lb-proof}. Run any algorithm for $T$ queries and split its recorded state at the weight threshold $dB/2$. If $T$ is below a small constant times $B\sqrt{N}$, the part above the threshold carries amplitude $2^{-\Omega(B)}$ by the progress bound, and the part below the threshold passes the success measurement with probability $e^{-\Omega(B)}$ by soundness. The triangle inequality then caps the overall success probability strictly below $1/3$, which proves \Cref{thm:bhvf}. The same template, run with a single recorded coordinate per block, proves the RMIF lower bound (\Cref{thm:rmif}) in \Cref{sec:appnd:RMIF-lb}.
\subsubsection{Reductions from Graph Problems to Intermediate Search Problems}
\paragraph{Triangle Listing}
The hard family starts from a bipartite Ruzsa--Szemer\'edi graph whose edges
split into $B=\Theta(n)$ large induced matchings~\cite{behrend1946sets,ruzsa1978triple}.  Each
matching represents one $\rmif$ block.  A small public augmentation turns the
marked edge of that block into one triangle, while inducedness prevents
unintended triangles.  Thus every reported triangle identifies the marked
location in a distinct block.  The adjacency, degree, and ordered-neighbor
oracles of the resulting graph are all simulated with constant overhead, so
\Cref{thm:informal-search} transfers to the graph problem. The detailed reduction is given in \Cref{sec:triangle-listing-reduction}.
\paragraph{k-Spanner} The hard instances have a different structure from dense high-girth graphs.
In fact, they have girth exactly four and each instance
admits a multiplicative $3$-spanner with only $O(n)$ edges. Thus the lower bound does not
arise from forcing the algorithm to print a dense answer, and the construction
does not seek a new extremal graph of large girth.  The difficulty is to locate
which input edges must be retained. 
The reduction starts from a sparse public bipartite skeleton.  Each skeleton
edge carries one $\bhvf$ block and contributes one hidden graph edge between
two coordinate arrays.  Public degree-balancing gadgets make all degrees
independent of the hidden pair.  The construction ensures that a constant
fraction of the hidden edges have no replacement path of length at most $k$;
these edges must occur in every valid $k$-spanner.  Outputting the spanner
therefore recovers the corresponding $\bhvf$ pairs.  The two graph endpoints
expose exactly the two source-query directions, and the full mixed graph oracle
is simulated with constant overhead.  The lower bound of \Cref{thm:informal-search} translates here. The detailed reduction is given in \Cref{sec:k-spanner-reduction}.

\section{Preliminaries}\label{sec:preliminaries}
\paragraph{Basic Notations.} Given any integers $1\leq m\leq n$, we define the set $[n]=\{0,1,2,...,n-1\}$ and $[m,n]=\{m,...,n-1\}$. Given a set $S$, we use $|S|$ to denote the cardinality of $S$. Given an array $A$, we use $|A|$ to denote the size of the array and we use $A[x]$ to denote the $x$-th element or $x$-th position of $A$ where $x\in [|A|]$. 
\paragraph{Notations for Quantum Computing.}
We use $\omega_N = e^{2i\pi/N}$ to denote a primitive $N$-th root of unity.
Bra-ket notation is used throughout: let $\ket{\psi}$ be a normalized (column) quantum state in a Hilbert Space $\CalH$, $\bra{\psi}$ be the conjugate transpose of $\ket{\psi}$, $\braket{\phi|\psi}$ be the inner product between $\ket{\phi}$ and $\ket{\psi}$, and $\ket\phi\bra\psi$ be the outer product between $\ket{\phi}$ and $\ket{\psi}$.

For an $n$ dimension Hilbert space the computational basis is $\{\ket{0},\ket{1},\cdots,\ket{n-1}\}$. The norm $\|\ket{\psi}\|$ refers to the standard $\ell_2$ norm.
For an orthonormal basis set $\{\ket{v_i}\}_i$ of an Hilbert space $\CalH$, if $\ket{\psi}=\sum_i \alpha_i\ket{v_i}$, then probability of $\ket{\psi}$ being measured as $\ket{v_i}$ is $|\alpha_i|^2$. A projector $\Pi$ on $\CalH$ is a PSD opeartor satisfying $\Pi^2=\Pi$. For a collection of orthogonal projectors $\{\Pi\}_i$ satisfying $\sum_i \Pi_i=\Id$, probability that a state $\ket\psi$ gives the measurement outcome $i$ is given by $\norm{\Pi_i\ket\psi}^2$. The range of a projector is $\Ran(\Pi)=\{\ket{x}\in \mathcal{H}: \Pi \ket{x} =\ket{x}\}$.
\\~\\
For a linear operator $A: \mathcal{H}\rightarrow \mathcal{H}$, the spectral operator norm is defined as 
\[
\norm{A}=\sup_{\ket{v}\in \mathcal{H}:\norm{\ket{v}}=1} \norm{A\ket{v}}.
\]
It has the property $\norm{A}^2=\norm{A^\dagger A}=\lambda_{\max}(A^\dagger A)$ where $\lambda_{\max}$ denotes the largest eigenvalue in magnitude. Note that the operator norm measures the largest possible amplification of vector length under the action of $A$. Given a positive semidefinite operator $A\geq 0$, it has $\norm{A}=\lambda_{\max}(A)$.
\paragraph{Notations for Graph}
Consider a graph $G:=(V,E)$ with the vertex set $V$ and the edge set $|E|$. For each vertex $v\in V$, we use $\mathfrak{N}(u)$ to denote the neighborhood list of the vertex $v$, which has an arbitrary and fixed ordering of neighbors of $v$. Let $\deg(v)=|\mathfrak{N}(v)|$ be the degree of vertex $v$.
For arbitrary vertices $u,v\in V$ and an integer $x\in \mathbb{N}_{\geq 0}$, 
\begin{itemize}
    \item We denote by $\mathcal{A}(u,v)$ the adjacency query which returns $1$ if there exists an edge $(u,v)\in E$ and returns $0$ otherwise.
    \item We denote by $\mathcal{D}(v)$ the degree query 
    which returns the vertex degree $\deg(v)$. 
    \item We denote by $\mathcal{N}(u,x)$ the neighbor query which returns the $x$-th neighbor in the neighborhood list of $u$ if $x\leq \deg(u)-1$ and returns a special symbol $\perp$ if $x\geq \deg(u)$. 
\end{itemize}
In this paper, an algorithm may coherently call any of the adjacency queries, degree queries and neighbor queries defined as above i.e.,
\begin{align*}
\mathcal{O}_{\mathcal{A}}\lvert u,v,z\rangle&=\lvert u,v,z\oplus \ind[\{u,v\}\in E]\rangle.\\
\mathcal{O}_{\mathcal{D}}\lvert v,z\rangle&=\lvert v,z\oplus \deg(v)\rangle,\\
\mathcal{O}_{\mathcal{N}}\lvert v,x,z\rangle&=\lvert v,x,z\oplus \mathfrak{N}(v)[x]\rangle
\end{align*}
An invalid port returns a public symbol $\perp$. One call to any oracle or its inverse counts as one
graph query.

\subsection{Quantum Algorithms and Quantum Query Models}
\label{subsec:quantum-query-model}
We adopt the standard formalism of quantum query complexity used in \cite{Ambainis00,BW02,Zhandry19,HamoudiMagniez2023,hamoudi2025briefintroductionquantumquery}, wherein an
algorithm interacts with an unknown function $f : [M] \to [N]$ exclusively
through quantum oracle calls, with the goal of producing an output
$w_{\mathrm{out}}$ satisfying a given relation $R(f, w_{\mathrm{out}})$.
\paragraph{Query bounded Quantum Algorithm.}
A $T$-query algorithm is given by a sequence of unitaries
$U_0, U_1, \ldots, U_T$ and oracle calls $\CalO_f$ acting on a Hilbert space $\CalH_{Alg}$ with three registers: the
query register $Q$ ranging over $[M]$, the phase register $P$ ranging over
$[N]$, and the working register $W$. We write basis states as
$\ket{x, u, w}_{QPW}$. Upon making $t \leq T$ queries to a function
$f : [M] \to [N]$, the algorithm's state evolves to
\begin{equation}\label{eqn:state-evolution-general}
    \ket{\psi^f_t}
    \;=\;
    U_t \mathcal{O}_f\, U_{t-1} \cdots U_1 \mathcal{O}_f\, U_0\, \ket{0}
\end{equation}
where $\mathcal{O}_f$ is the phase oracle defined by its action on basis states:
\begin{equation}
    \mathcal{O}_f \ket{x, u, w}
    \;=\;
    \omega_N^{\,u f(x)}\, \ket{x, u, w},
    \qquad
    \omega_N := e^{2i\pi/N}.
\end{equation}
The value $f(x)$ is encoded into the phase of the quantum amplitude rather
than written into an output register. Although an alternative convention writes
the oracle as the map $\ket{x, u, w} \mapsto \ket{x, u + f(x) \bmod N, w}$,
the two are unitarily equivalent and the phase encoding is more amenable to the query framework we employ. The algorithm's answer is given by a
designated substring $w_{\mathrm{out}}$ of $w$. Writing
$\Pi^f_{\mathrm{succ}}$ for the projector onto basis states $\ket{x, u, w}$
whose output substring satisfies $R(f, w_{\mathrm{out}})$, the
\emph{success probability} on input $f$ is
\begin{equation}
    \sigma_f \;=\; \bigl\|\Pi^f_{\mathrm{succ}}\ket{\psi^f_T}\bigr\|^2.
\end{equation}

\paragraph{Oracle Register.}\label{para:oracle-register} As described in \cite{Ambainis00, Zhandry19, HamoudiMagniez2023}
it is often convenient to maintain the input function explicitly in an
additional \emph{oracle register} $F$, which evolves in superposition
alongside the algorithm. This perspective is standard in the adversary method
and is central to the recording technique used in \cite{Zhandry19, HamoudiMagniez2023}. The register $F$ is decomposed
into $M$ sub-registers $F_1, \ldots, F_M$, with $F_x$ storing the value
$f(x) \in [N]$ for each $x \in [M]$. Basis states of the oracle Hilbert space (denoted by $\CalH_{inp}$)
are written $\ket{f}_F := \bigotimes_{x \in [M]} \ket{f(x)}_{F_x}$.

Given a distribution $\mathcal{D}$ over the function space $[N]^M$, the
initial state of the oracle register is set to
\begin{equation}\label{eqn:Separable-input-oracle}
    \ket{\mathcal{D}}_F := \sum_{f \in [N]^M} \sqrt{\Pr[f \leftarrow \mathcal{D}]}\,\ket{f}_F
\end{equation}
The joint query operator $\mathcal{O}$ acts on both the algorithm workspace
and $F$ by
\begin{equation}
    \mathcal{O}\ket{x, u, w}\ket{f}
    \;=\;
    \bigl(\mathcal{O}_f\ket{x, u, w}\bigr)\ket{f},
\end{equation}
keeping $F$ intact. Extending each $U_i$ trivially to act as the identity on
$F$, the joint state of the algorithm and oracle after $t$ queries becomes
\begin{equation}\label{eqn:state-evolution-separable-input}
    \ket{\psi_t}
    \;=\;
    U_t \mathcal{O}\, U_{t-1} \cdots U_1 \mathcal{O}\, U_0
    \bigl(\ket{0}\ket{\mathcal{D}}\bigr)
    \;=\;
    \sum_{f \in [N]^M}
    \sqrt{\Pr[f \leftarrow \mathcal{D}]}\;\ket{\psi^f_t}\ket{f}.
\end{equation}
The overall success probability of the algorithm under the distribution
$\mathcal{D}$ is then measured by jointly observing the working and oracle
registers of $\ket{\psi_T}$. Letting $\Pi_{\mathrm{succ}}$ be the projector
onto basis states $\ket{x,u,w}\ket{f}$ for which $w_{\mathrm{out}}$ satisfies
$R(f, w_{\mathrm{out}})$, this probability is
\begin{equation}
    \sigma \;=\; \bigl\|\Pi_{\mathrm{succ}}\ket{\psi_T}\bigr\|^2.
\end{equation}
The graph reductions used later are classical in their description.  In the
quantum query model, however, the corresponding oracle simulation must be
performed coherently and must erase all temporary source-oracle information.

\begin{fact}[Implicit in Section 8 of \cite{DHHM06}]
\label{fact:clean-coherent-graph-reduction}
Let $\mathsf{P}_{\mathrm{src}}$ be an oracle search problem whose input
$x$ is sampled from a distribution $\mathcal{D}_{\mathrm{src}}$, and let
$\mathcal{O}_x$ denote its query oracle.  Suppose that a public construction associates
with every source input $x$ a graph $G_x$, and let
$\mathcal{D}_{\mathrm{gr}}$ denote the induced distribution of $G_x$ when
$x\sim\mathcal{D}_{\mathrm{src}}$.  Assume that the following properties hold.

\begin{enumerate}
    \item Every valid output for the graph instance $G_x$ can be converted,
    without any additional calls to $\mathcal{O}_x$, into a valid output for the source
    instance $x$.

    \item Every adjacency, degree, or ordered-neighbor query to $G_x$ can be
    answered exactly by a deterministic classical procedure using at most
    $c=O(1)$ calls to $\mathcal{O}_x$.
\end{enumerate}

Then any $T$-query quantum algorithm for the graph problem can be converted
into an $O(cT)$-query quantum algorithm for
$\mathsf{P}_{\mathrm{src}}$, with at least the same success probability.
Consequently, a distributional lower bound $\Omega(L)$ for
$\mathsf{P}_{\mathrm{src}}$ under $\mathcal{D}_{\mathrm{src}}$ implies an
$\Omega(L)$ quantum query lower bound, up to the constant simulation overhead, for the graph problem under
$\mathcal{D}_{\mathrm{gr}}$.  It also implies the same lower bound for the
worst-case graph problem on the corresponding graph family.
\\~\\
The classical query procedure is implemented reversibly in
compute--write--uncompute form.  Namely, it coherently computes the required
source-oracle answers into work registers, computes the graph-oracle response,
writes that response into the designated answer register or applies the
corresponding phase and then reverses the entire computation.  Thus all
temporary registers are returned to zero.
\end{fact}
\subsection{Recording Query Model}
\label{subsec:recording-query-model}
The quantum recording query model is a modification of the standard query model
of Section~\ref{subsec:quantum-query-model} designed to be transparent to the algorithm while
simultaneously enabling a fine-grained accounting of the algorithm's progress
towards solving the problem at hand. The original formulation is due to \cite{Zhandry19}, but we adopt here the version of \cite{HamoudiMagniez2023} which is a simplified and more general variant that
applies to any product distribution on the input, rather than being restricted
to the uniform distribution.
\paragraph{Construction.}
Throughout this section we fix a product distribution
$\mathcal{D} = \mathcal{D}_1 \otimes \cdots \otimes \mathcal{D}_M$ over $[N]^M$.
The starting point of the construction is to enlarge the range of the oracle
from $[N]$ to $[N] \cup \{\perp\}$, so that the oracle register may now hold
functions $f : [M] \to [N] \cup \{\perp\}$. This extension of the space $\CalH_{inp}$ is denoted as $\CalH_{rec}$ with dimension $(N+1)^M$. The symbol $\perp$ at position $x$
signals that the algorithm has not yet obtained any information about $f(x)$.
By convention, any basis state $\ket{x,u,w}\ket{f}$ in which $f$ takes the
value $\perp$ at any coordinate lies outside the support of
$\Pi_{\mathrm{succ}}$, and hence cannot contribute to the success probability.
Crucially, the oracle register is now initialised independently of the
distribution $\mathcal{D}$, always starting in the fully unknown state
$\ket{\perp^M}_F$. The standard query operator $\mathcal{O}$ is extended
consistently by declaring it to act as the identity on states with $f(x) = \perp$.

To connect the recording model back to the standard one, we introduce for each
$x \in [M]$ the state
$\ket{\mathcal{D}_x}_{F_x} := \sum_{y \in [N]} \sqrt{\Pr[y \leftarrow \mathcal{D}_x]}\,\ket{y}_{F_x}$
and a unitary $S_x$ acting on register $F_x$ that swaps $\ket{\perp}$ and
$\ket{\mathcal{D}_x}$ while fixing every vector orthogonal to both:
\begin{equation}\label{eqn:action-of-S-x}
    S_x :
    \begin{cases}
        \ket{\perp}        \;\mapsto\; \ket{\mathcal{D}_x}, \\
        \ket{\mathcal{D}_x} \;\mapsto\; \ket{\perp}, \\
        \ket{\varphi}      \;\mapsto\; \ket{\varphi}
        \quad \text{for } \ket{\varphi} \perp \operatorname{span}\{\ket{\perp},\,\ket{\mathcal{D}_x}\}.
    \end{cases}
\end{equation}
From these local operators we build two global unitaries on
$\mathcal{H}_{QPW} \otimes \mathcal{H}_F$:
\begin{align}
    \CalT_{\mathcal{D}}
    &\;=\; \Id_{QPW} \otimes \bigotimes_{x \in [M]} S_x,
    \label{eq:TD} \\
    S_{\mathcal{D}}
    &\;=\; \sum_{x \in [M]} \ket{x}\bra{x}_Q \otimes
           \Id_{PW} \otimes
           \Id_{F_1 \cdots F_{x-1}} \otimes S_x \otimes \Id_{F_{x+1} \cdots F_M}.
    \label{eq:SD}
\end{align}
The operator $\CalT_{\mathcal{D}}$ applies $S_x$ to every sub-register simultaneously,
while $S_{\mathcal{D}}$ applies $S_x$ only to the sub-register $F_x$ selected
by the current query index in $Q$. The \emph{recording query operator} is then
defined as
\begin{equation}
    \CalR_{\mathcal{D}} \;:=\; S_{\mathcal{D}}^\dagger\, \mathcal{O}\, S_{\mathcal{D}}.
    \label{eq:RD}
\end{equation}
The joint state of the algorithm and oracle after $t$ steps in the recording
model is
\begin{equation}\label{eqn:recording-evolution-separable}
    \ket{\phi_t}
    \;=\;
    U_t \CalR_{\mathcal{D}}\, U_{t-1} \cdots U_1 \CalR_{\mathcal{D}}\, U_0
    \bigl(\ket{0}\ket{\perp^M}\bigr).
\end{equation}
Since $\CalR_{\mathcal{D}}$ can modify sub-register $F_{x'}$ only when the query
index satisfies $x = x'$, the following structural fact is immediate.

\begin{fact}[Fact 3.2 from \cite{HamoudiMagniez2023}]
\label{fact:support}
    For every $t \geq 0$, the state $\ket{\phi_t}$ is supported on basis states
    $\ket{x,u,w}\ket{f}$ in which $f$ has \emph{at most $t$ entries
    different from $\perp$}.
\end{fact}

Intuitively, $\ket{\phi_t}$ keeps a running record of which coordinates have
been queried: at most $t$ coordinates can carry a non-$\perp$ value after
$t$ oracle calls.

\paragraph{Indistinguishability.}
The recording model would be of little use if it altered the behavior of the
algorithm. The following theorem, which is the cornerstone of the entire
framework, shows that the two models produce identical reduced states on the
algorithm's registers.

\begin{theorem}[Theorem 3.3 from \cite{HamoudiMagniez2023}]
\label{thm:indistinguishable}
    Let $\mathcal{D}$ be a product distribution and let $(U_0, \ldots, U_T)$
    be a $T$-query quantum algorithm. Let $\ket{\psi_T}$ and $\ket{\phi_T}$
    denote the final states in the standard and recording query models
    respectively. Then
    \begin{equation}
        \ket{\psi_T} \;=\; \CalT_{\mathcal{D}}\,\ket{\phi_T}.
        \label{eq:indistinguishable}
    \end{equation}
    In particular, the marginal states on the algorithm registers $QPW$ are
    identical in both models.
\end{theorem}
As an immediate consequence, the success probability in the standard model can
be expressed entirely in terms of the recording model state:
\begin{equation}
    \sigma
    \;=\; \bigl\|\Pi_{\mathrm{succ}}\ket{\psi_T}\bigr\|^2
    \;=\; \bigl\|\Pi_{\mathrm{succ}}\,\CalT_{\mathcal{D}}\,\ket{\phi_T}\bigr\|^2.
    \label{eq:sigma-recording}
\end{equation}
Any upper bound on $\sigma$ established in the recording model therefore
transfers without loss to the standard model.
\section{Intermediate Searching Problems}
\label{sec:two-problems}
\subsection{Bidirectional Hidden-Value Finding (BHVF) Problem}\label{subsec:BHVF}
Let $B, N \geq 4$ be some given sufficiently large parameters. Consider an input $r=(r_t)_{t \in [B]}$ where $r_t=(r_t^A, r_t^B)$ is a pair with $r_t^A, r_t^B \in [N]$ sampled independently and uniformly. Set up $B$ blocks. For each block $t \in [B]$, we define two functions $A^r_t, B^r_t: [N] \rightarrow [N+1]$ such that
\[
    A^r_t(x)=
    \begin{cases}
        r^B_t+1, &x=r^A_t,\\
        0, &x \neq r^A_t,
    \end{cases}
    \qquad
    B^r_t(x)=
    \begin{cases}
        r^A_t+1, &x=r^B_t,\\
        0, &x \neq r^B_t.
    \end{cases}
\]
We call $(r^A_t,r^B_t)$ the hidden values for block $t$. To find the hidden values, the (classical or quantum) algorithm can query the $\bhvf$ oracle $O_r$ such that given a block index $t\in [B]$, a type $\tau \in \{A,B\}$ and a position $x\in [N]$, 
\[
\text{if } \tau = A, \quad O_r(t, \tau, x) = \begin{cases}
    r_t^B+1, \quad & x=r_t^A,\\
    0, & x \neq r_t^A.
\end{cases} 
\]
\[\text{if } \tau = B, \quad 
O_r(t, \tau, x) = \begin{cases}
    r_t^A+1, \quad & x=r_t^B, \\
    0, & x \neq r_t^B.
\end{cases} 
\]

\begin{problem}[$(B,N,d)$-$\bhvf$]\label{prob:mbeh}
    Given sufficiently large integers $B,N$ and some constant $d \in (0,1]$, the task $(B,N,d)$-$\bhvf$ is to return the hidden values $(r_t^A, r^B_t)$ for at least $\lceil dB\rceil$ distinct blocks.

    In the quantum query model, the input graph $r=(r_t^A,r_t^B)_{t\in [B]}$ with $r_t^A , r_t^B\in [N]$ sampled uniformly and independently. The input can be accessed by the $\bhvf$ oracle. 
    The success probability of the algorithm should be at least $2/3$ over both the randomness of the input and the randomness of the algorithm.
\end{problem}

\begin{theorem}\label{thm:main-mbeh short version}\label{thm:bhvf}
    Given any constant $d\in (0,1]$, the task $(B,N,d)$-$\bhvf$ has quantum query complexity $\Omega(B \sqrt{N})$.
\end{theorem}
A simple one-directional analogue of this problem is the following.
\subsection{Restricted Marked-Item Finding (RMIF) Problem}\label{subsec:RMIF}
Let $B, N \geq4$ be some given sufficiently large parameters. Consider an input $r=(r_t)_{t \in [B]}=(r_0, r_1, \dots, r_{B-1})$ where each $r_t\in [N]=\{0,1,\dots, N-1\}$ is sampled independently and uniformly. Set up $B$ blocks. For each block $t \in [B]$, define a function $f^r_t:[N] \rightarrow \{0,1\}$ such that 
\[
    f^r_t(x)=
    \begin{cases}
        1,&x=r_t,\\
        0,&x\ne r_t.
    \end{cases}
\]
We call $r_t$ the marked item for block $t$. To find the marked items, the (classical or quantum) algorithm can query the $\bhvf$ oracle $O_r$ such that given a block index $t\in [B]$ and a position $x\in [N]$, 
\[
O_r(t, x) = \begin{cases}
    1, \quad & x=r_t,\\
    0, & x \neq r_t.
\end{cases} 
\]

\begin{problem}[$(B,N,d)$-$\rmif$]
    Given sufficiently large integers $B,N$ and some constant $d \in (0,1]$, the task $(B,N,d)$-$\rmif$ is to return the hidden values $r_t$ for at least $\lceil dB\rceil$ distinct blocks.

    In the quantum query model, the input graph $r=(r_t)_{t\in [B]}$ with $r_t\in [N]$ sampled uniformly and independently. The input can be accessed by the $\rmif$ oracle. 
    The success probability of the algorithm should be at least $2/3$ over both the randomness of the input and the randomness of the algorithm.
\end{problem}

\begin{theorem}\label{thm:rmif}
    Given any constant $d\in (0,1]$, the task $(B,N,d)$-$\rmif$ has quantum query complexity $\Omega(B \sqrt{N})$.
\end{theorem}

\section{Lower Bounds of the Triangle-Listing Problem}\label{sec:triangle-listing-reduction}
In this section, we prove the lower bounds for the triangle listing problem .
\begin{problem}[$(n,t,d)$-\trianglelisting]
Given some constant $d \in (0,1]$ and an $n$-node graph class $\mathcal{G}_n$ such that every graph $G\in \mathcal{G}_n$ contains exactly $t$ triangles, randomly sample a graph $G$ from $\mathcal{G}_n$ and the task is to list at least $\lceil dt\rceil$ triangles in $G$. 

In the general graph query model, the sampled graph can be accessed by adjacency queries, neighborhood queries and degree queries. The success probability of the algorithm should be at least $2/3$ over both the randomness of the input and the randomness of the algorithm.
\end{problem}
\begin{theorem}\label{thm:triangle listing main}
    Given some $t=\Theta(n)$, for any constant $d \in (0,1]$, the  $(n,t,d)$-$\trianglelisting$ problem requires $\Omega\left(\frac{n^{3/2}}{2^{O(\sqrt{\log n})}}\right)$ i.e., $\Omega(n^{3/2-o(1)})$ quantum queries.
\end{theorem}
To prove \Cref{thm:triangle listing main}, we introduce the definitions of induced matchings and Ruzsa-Szemeredi graphs in \Cref{subsec:induced-matching-RS-graph} to construct graph instances and prove a reduction in \Cref{subsec:reduction-from-RMIF-tri-listing}.

\subsection{Induced matchings and Ruzsa-Szemeredi graphs}\label{subsec:induced-matching-RS-graph}

\begin{definition}[Induced Matching] Given a graph $G=(V,E)$, we say a set of edges $E'\subseteq E$ is an induced matching of $G$ if $E'$ is a matching and there is a node subset $S \subseteq V$ such that $E'$ is exactly the edge set of the subgraph induced by $S$.
\end{definition}
\begin{definition}
    Given some integers $r,t\geq 1$, we call a graph $G$ an $(r,t)$-Ruzsa-Szemeredi graph if there exists a partition of its edges into $t$ sets of size $r$, such that each set constitutes an induced matching of $G$.
\end{definition}
These graphs were first introduced in the famous paper by Ruzsa and Szemeredi \cite{ruzsa1978triple}, in which they prove that there exists no $n$-node $(r,t)$-RS graph $G=(V,E)$ with $r =\Omega(n)$ and $t =\Omega(n)$. In addition, they constructed the following result, based on the result of Behrend \cite{behrend1946sets}.

\begin{theorem}[ \cite{behrend1946sets,ruzsa1978triple}]\label{thm:rt RS-graph}
    There exist $n$-node $(\frac{n}{2^{O(\sqrt{\log n})}},\frac{n}{3})$-Ruzsa-Szemeredi graphs. 
\end{theorem}
In this paper, we utilize the  bipartite $(r,t)$-RS graph with $r=\frac{n}{2^{O(\sqrt{\log n})}}$ and $t=\frac{n}{3}$ constructed from $\Cref{thm:rt RS-graph}$.
\begin{lemma}\label{cor:base graph}
    There exists a bipartite graph $G_{RS}=(L,R,E_{RS})$ with $|L|=|R|=n$ such that it is a $(\frac{n}{2^{O(\sqrt{\log n})}},\frac{n}{3})$-Ruzsa-Szemeredi graph.
\end{lemma}
\begin{proof}[Proof of \Cref{cor:base graph}]
    Let $G=(V,E)$ be the $n$-node $(r=\frac{n}{2^{O(\sqrt{\log n})}},t=\frac{n}{3})$-Ruzsa-Szemeredi graph by \Cref{thm:rt RS-graph}. Let $V=\{v_1, v_2, \dots, v_n\}$. Set up the node sets $L=\{u_1, \dots , u_n\}$ and $R=\{w_1, \dots , w_n\}$. Start from an empty set
    $E_{RS}$.
    For every edge $(v_i,v_j)\in E$, add the edges $(u_i,w_j)$ and $(u_j,w_i)$ to the set $E_{RS}$. Let $G_{RS}=(L,R,E_{RS})$ be the final graph. It is easy to see that $G_{RS}$ is a bipartite graph with $|L|=|R|=n$. In the following, we prove that $G_{RS}$ is a $(2r,t)$-RS graph where $r=\frac{n}{2^{O(\sqrt{\log n})}}$ and $t=\frac{n}{3}$.

    Let $E=E_0 \sqcup E_1 \sqcup \ldots \sqcup E_{t-1}$ be the partition of $E$ such that every $E_i\subseteq E$ is an induced matching of $G$ with size $r$.
    For every $k \in [t]$, set up the edge subset $E_k^*\subseteq E_{RS}$ such that \[
    (u_i, w_j),(u_j, w_i)\in E^*_k \iff (v_i, v_j) \in E_k.
    \]
    It is easy to verify that for every subset $E^*_k \subseteq E_{RS}$, it is an induced matching of $G_{RS}$ with size $|E^*_k|=2r$ and
    $E_{RS}=E^*_0\sqcup \ldots \sqcup E^*_{t-1}$. 
\end{proof}
\subsection{Reduction from RMIF to Triangle-Listing}\label{subsec:reduction-from-RMIF-tri-listing}
For every large integer $n\geq 100$, consider the parameters $B=n/3$ and $N=n/2^{O(\sqrt{\log n})}$ such that the $2n$-node bipartite $(r=N,t=B)$-Ruzsa-Szemeredi graph $G_{RS}=(L,R,E_{RS})$ exists with $|L|=|R|=n$. Fix the graph $G_{RS}$. Note that the edge set $E_{RS}$ can be partitioned as
\[
E_{RS}=E_0 \sqcup E_1 \sqcup \ldots \sqcup E_{B-1}
\]
where each $E_i$ is an induced matching of $G_{RS}$ with size $|E_i|=N$. We label nodes and edges in the following way.
\begin{enumerate}
    \item For each node $u\in L \cup R$, we use $d_{RS}(u)$ to denote its degree and $N_{RS}(u)= \{u_1, u_2,\dots, u_{k}\}$ to node its neighborhood with $k=d_{RS}(u)=|N_{RS}(u)|$. 
    \item For each induced matching $E_i$ of $G_{RS}$ with $i \in [B]$, label its edges and nodes such that 
    \[
    E_i = \{e^{(i)}_0, e^{(i)}_1,\ldots, e^{(i)}_{r-1} \}
    \]
    and every edge $e_j^{(i)}=(x^{(i)}_j, y^{(i)}_j)$ where $x^{(i)}_j \in L$ and  $y^{(i)}_j \in R$. 
    Let $L_i=\{x_j^{(i)}: i\in [N] \}$ and $R_i=\{y_j^{(i)}: i\in [N] \}$ be the node sets of $E_i$. 
\end{enumerate}
Therefore, for every edge $e\in E_{RS}$, we know the indices $j\in [N]$ and $i\in [B]$ such that $e= e^{(i)}_j$ is the $j$-th edge in $E_i$.

\paragraph{Graph Construction} Consider any input $r=(r_i)_{i\in [B]}$ of the RMIF problem where every $r_i \in [N]$. First, we construct a graph $G_r$ from $G_{RS}$ described by $r$. Go through every block $i\in [B]$ with input $r_i \in [N]$, 
\begin{enumerate}
    \item Add the edge $e^{(i)}_{r_i}=(x^{(i)}_{r_i}, y^{(i)}_{r_i}) \in E_i$.
    \item Add a gadget node $w_i$ and the edges $(w_i,u)$ for each node $u\in L_i \cup R_i$.
    \item Add a gadget node $w'_i$ and the edges $(w'_i, u)$ for each node $u \in L_i \cup R_i \setminus \{x^{(i)}_{r_i}, y^{(i)}_{r_i}\}$.
    \item Add two dummy nodes $s_i^x, s_i^y$ and two edges $(w'_i, s_i^x)$ and $(w'_i, s_i^y)$.
\end{enumerate}
See \Cref{fig:1 reduction to triangle finding} for an illustration.
Let $W=\{w_i, w'_i\}_{i \in [B]}$ be the set of gadget nodes and $S=\{s_i^x, s_i^y\}_{i \in [B]}$ be the set of dummy nodes. We call nodes in $L \cup R$ regular nodes.
For every possible input $r=(r_i)_{i\in [B]}$, we use $G_{r}=(V_r, E_r)$ to denote the graph described by $r$ where $V_r =L\cup R\cup W\cup S$ has size $|V_r|=2n + 4B=10n/3$. Every regular node $u \in L \cup R$ has degree $2d_{RS}(u)$ in $G_r$. The gadget node $w\in W$ has degree $2N$ and the dummy node $s \in S $ has degree $1$ in $G_r$.
\begin{observation}\label{obs:triangle-labels}
    For every $r=(r_i)_{i\in[B]}$, the graph $G_r$ contains exactly $B=n/3$ triangles over the nodes $\{x^{(i)}_{r_i}, y^{(i)}_{r_i}, w_i\}$ for every $i\in [B]$.
\end{observation}
\begin{figure}
    \centering
    \includegraphics[width=0.5\linewidth]{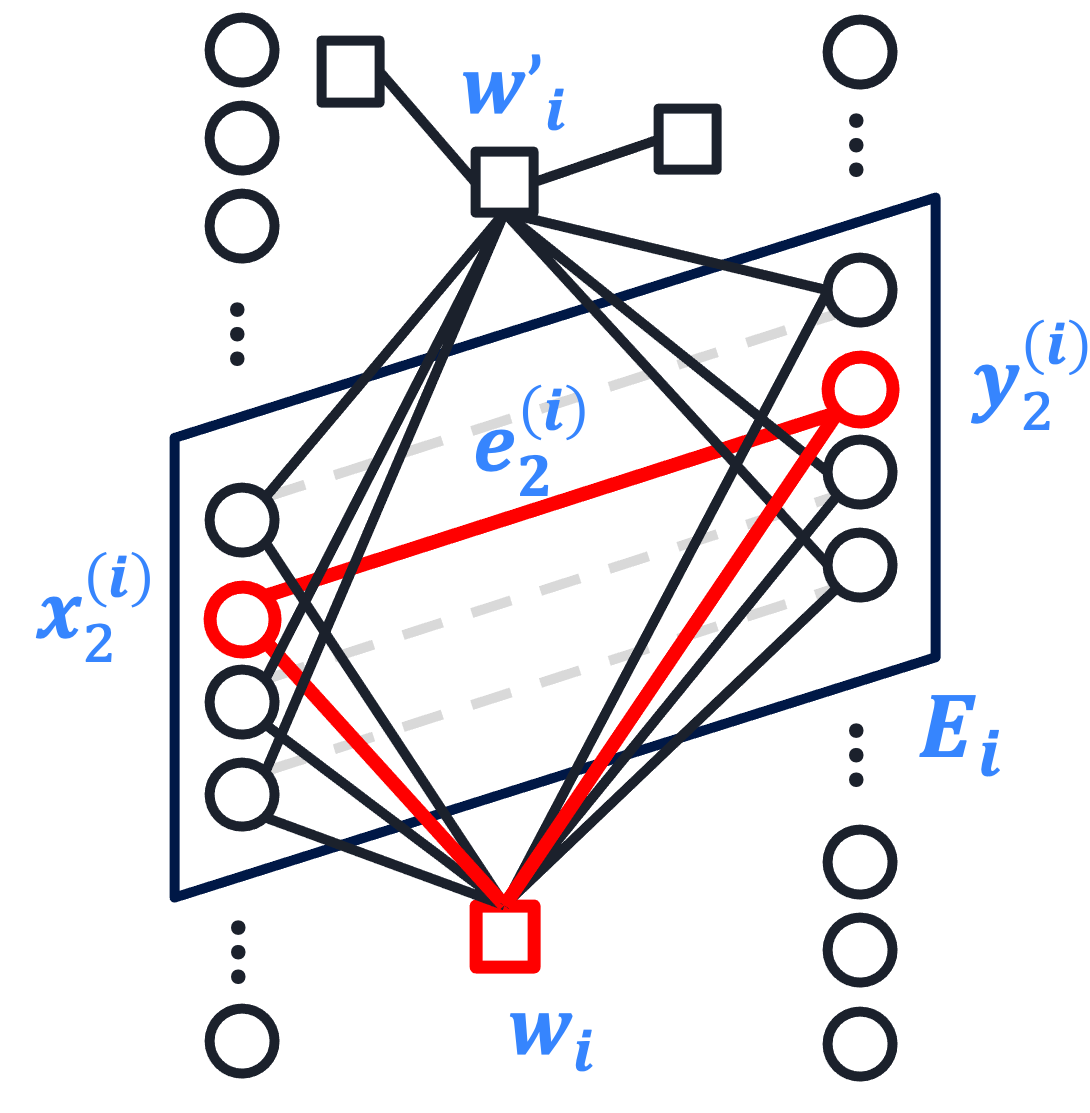}
    \caption{Consider any induced matching $E_i \subseteq E_{RS}$ with $i\in [B]$. The marked item for block $i$ has value $2$. All unselected edges of $E_i$ are represented by gray dashed lines.
    In addition, the triangle over $\{x_2^{(i)}, y_2^{(i)}, w_i\}$ is highlighted in red.
    }
    \label{fig:1 reduction to triangle finding}
\end{figure}
\paragraph{Reduction from RMIF} For every large integer $n\geq 100$, let $B=n/3$ and $N=n/2^{O(\sqrt{\log n})}$ be the chosen parameters and $\mathcal{G}=\{G_{r}\}$ be the constructed graph class. Note that $\mathcal{G}=\{G_{r}\}$ consists of the graph $G_r=(V_r, E_r)$ for every possible input $r=(r_t)_{t\in [B]}$. Every graph $G_r \in \mathcal{G}$ has size $|V_r|=10n/3$ and contains $B=n/3$ triangles.
We prove the following main result.

\begin{theorem}\label{thm:Reduction-RMIF-to-T-Listing}
    If there exists an algorithm $\mathcal{B}$ (classical or quantum) that solves $(\frac{10n}{3},\frac{n}{3}, d)$-$\trianglelisting$ over $\mathcal{G}$ using $Q$ graph queries, there exists an algorithm $\mathcal{A} $ (respectively classical or quantum) that solves $(B=\frac{n}{3},N=\frac{n}{2^{O(\sqrt{\log n})}},d)$-$\rmif$ using $O(Q)$ $\rmif$ oracle queries.
\end{theorem}

\begin{proof}
Let $\mathcal{B}$ be the given algorithm that solves $(10n/3,n/3, d)$-$\trianglelisting$ over $\mathcal{G}$ using $Q$ queries. Firstly, note that the algorithm $\mathcal{B}$ can access to adjacency queries, neighborhood queries and degree queries to the input graph.

Let all the degree queries be free, since the degrees in $G_r$, denoted by $\deg_{G_r}(u)$, are predetermined. In the following, we prove that every adjacency query and every neighborhood query over $\mathcal{G}$ can be simulated with $O(1)$ calls to the $\rmif$ oracle. The proof is explicitly for the classical version of the oracles, and it can be easily lifted to the quantum speedup by calling oracles in superposition since we prove it via case-by-case analysis.
\begin{description}
    \item[Adjacency Queries] Given any two input nodes $u,v\in V_r$, we first categorize $u,v$ in the categories of the regular nodes, the gadget nodes, and the dummy nodes. 
    Now we proceed via the following cases:
    \begin{enumerate}
        \item Assume (w.l.o.g.) $u\in \{s_i^x, s_i^y\}$ for some $i\in [B]$. Return $\mathcal{A}(u,v)=1$ if and only if $v=w'_i$.
        \item Assume (w.l.o.g.) $u=w_i$ for some $i\in [B]$. Return $\mathcal{A}(u,v)=1$ if and only if $v \in L_i \cup R_i$.
        \item Assume (w.l.o.g.) $u=w'_i$ for some $i\in [B]$. In addition, assume $v \in L_i \cup R_i$, since otherwise we can return $\mathcal{A}(u,v)=0$.
        Let $j \in [N]$ be the index such that $v=x^{(i)}_j$ (w.l.o.g.). Query $O_r(i, j)$ and return $\mathcal{A}(u,v)=1$ if and only if $O_r(i, j)=0$.
        \item In the last case, both $u,v$ are regular nodes. If $(u,v)\notin E_{RS}$, then $\mathcal{A}(u,v)=0$. Otherwise, let $i \in [B]$ and $j\in [N]$ be the indices such that $(u,v)=e^{(i)}_j\in E_i$. Query $O_r(i, j)$ and return $\mathcal{A}(u,v)=1$ if and only if $O_r(i, j)=1$.
    \end{enumerate}
    \item[Neighborhood Queries] Given an input node $u\in V_r$ and an integer $k \in \mathbb{N}_{\geq 0}$, we assume $k \in [\deg_{G_r}(u)]$ since $\deg_{G_r}(u)$ is given for free. For each regular node $u\in L\cup R$ in $G_{RS}$, we publicly order its neighbors such that $N_{RS}(u)=\{u_0,u_1, \ldots, u_{d_{RS}(u)-1}\}$. For each neighbor $u_q \in N_{RS}(u)$, let $i_q \in [B]$ and $j_q \in [N]$ be the public indices such that
    \[
    e^{(i_q)}_{j_q} = (u,u_q)\in E_{i_q}, \quad \forall q\in [d_{RS}(u)].
    \]

    \begin{itemize}
        \item For the adjacency list of a regular node $u\in L \cup R$, it has length $2 d_{RS}(u)$. Given $k\in [d_{RS}(u)]$, set $q=k$ and return 
        \[
        \mathcal{N}(u,k=q)=w_{i_k}.
        \]
        For other input $k\in [d_{RS}(u), 2d_{RS}(u)]$, set $q=k-d_{RS}(u)$ and compute $O_r(i_q, j_q)$. Return
        \begin{equation*}
            \mathcal{N}(u,k=d_{RS}(u)+q)=\begin{cases}
                 u_{q} \quad & \text{if }  O_r(i_q, j_q)=1, \\
                 w'_{i_q} \quad & \text{if } O_r(i_q, j_q)=0.
             \end{cases}
        \end{equation*}
        \item For the adjacency list of a gadget node $u=w_i \in W$ for some block $i\in [B]$, it has length $2N$. Return
        \[
        \mathcal{N}(u,k)=\begin{cases}
            x^{(i)}_k \quad & \text{if } k \in [N],  \\
            y^{(i)}_{k-N} & \text{if } k \in [N,2N].
        \end{cases}
        \]
        \item For the adjacency list of a gadget node $u=w'_i \in W$ for some block $i\in [B]$, it has length $2N$. Given $k \in [N]$, query $O_r(i, k)$ and return
        \[
        \mathcal{N}(u,k)=\begin{cases}
            x^{(i)}_k \quad & \text{if } O_r(i, k)=0,  \\
            s_i^{x} & \text{if } O_r(i, k)=1.
        \end{cases}
        \]
        Given $k \in [N,2N]$, query $O_r(i, k-N)$ and return
        \[
        \mathcal{N}(u,k)=\begin{cases}
            y^{(i)}_{k-N} \quad & \text{if } O_r(i, k-N)=0,  \\
            s_i^{y} & \text{if } O_r(i, k-N)=1.
        \end{cases}
        \]
        \item For the adjacency list of a dummy node $u\in \{s_i^x, s_i^y\}$ for some block $i\in [B]$, it has length $1$. Return $\mathcal{N}(u, k)=w'_i$ with $k=0$.
    \end{itemize}
    Therefore, every adjacency query and every neighborhood query over $\mathcal{G}$ can be simulated with $O(1)$ calls to the $\rmif$ oracle.

    Given a random input $r=(r_t)_{t\in [B]}$, we construct the corresponding graph $G_r$ described by $r$. By \Cref{obs:triangle-labels}, there are exactly $B=n/3$ triangles and every triangle is over the node set $\{w'_i, x^{(i)}_{r_t}, y^{(i)}_{r_t}\}$ for $i \in [B]$. 
    
    We can simulate the algorithm $\mathcal{B}$ that solves $(10n/3,n/3,d)$-$\trianglelisting$ over $G_r$. For every triangle over the node $\{w'_i, u, v\}$ returned by the algorithm $\mathcal{B}$, it implies that $r_i = j$ where $j\in [N]$ is the index such that
    edge $(u,v)=e^{(i)}_j$.

    Therefore, given the output of the algorithm $\mathcal{B}$ with at least $dB$ distinct triangles in $G_r$, we can return the marked item $r_i \in [N]$ for at least $dB$ distinct blocks.
\end{description}
\end{proof}
\paragraph{Proof of \Cref{thm:triangle listing main}} Now we complete the proof of \Cref{thm:triangle listing main}.

For every large integer $n\geq 100$, let $B=n/3$ and $N=n/2^{O(\sqrt{\log n})}$ be the chosen parameters and
$\mathcal{G}=\{G_{r}\}$ be the constructed graph class. 

Consider any algorithm $\mathcal{B}$ that solves $(10n/3,n/3,d)$-$\trianglelisting$ over $\mathcal{G}$ with $Q$ quantum graph queries. By \Cref{thm:Reduction-RMIF-to-T-Listing}, the task $(B,N,d)$-$\rmif$ can be solved with $Q'=O(Q)$ 
quantum queries to the $\rmif$ oracle. 
By 
\Cref{thm:rmif}, $Q'=\Omega(B\sqrt{N})$. Therefore
\[
Q=\Omega(B \sqrt{N}) = \Omega\left(n\cdot \left(\frac{n}{2^{O(\sqrt{\log n})}}\right)^{1/2}\right)
=
\Omega\left(\frac{n^{3/2}}{2^{O(\sqrt{\log n})}}\right).
\]

\section{Lower Bounds for the k-Spanner Problem}\label{sec:k-spanner-reduction}
In this section, we prove \Cref{thm:generic-lower-bound} for the k-spanner construction.

\begin{problem}[$(n,k)$-$\spanner$]
Given some integer $k \geq 1$ and an $n$-node graph class $\mathcal{G}_n$, randomly sample a graph $G$ from $\mathcal{G}_n$ and the task is to output an edge set
$E' \subseteq E$ such that the subgraph $H=(V,E')$ is a $k$-spanner of $G$.

In the general graph query model, the sampled graph can be accessed by adjacency queries, neighborhood queries and degree queries.
The success probability of the algorithm should be at least $2/3$ over both the randomness of the input and the randomness of the algorithm.
\end{problem}

\begin{theorem}\label{thm:generic-lower-bound}
For every fixed integer $k\ge 7$, the $(n,k)$-$\spanner$ has the quantum query complexity
\[
\Omega\!\left(n^{1+\frac{1}{2\mu_k}}\right)
\]
where $\mu_k=\left\lfloor(k+1)/3\right\rfloor
-\mathbbm{1}[k\equiv 2\pmod 6]$.
\end{theorem}
Note that $\mu_7 = \mu_8 = 2$. For large integer $k$, $\mu_k = k/3+O(1)$. We have the following corollaries.

%

\begin{corollary} 
The $(n,k=7)$-$\spanner$ and  $(n,k=8)$-$\spanner$ has the quantum query complexity $\Omega(n^{5/4})$.
\end{corollary}
\begin{corollary}
    For sufficiently large integer $k$, the $(n,k)$-$\spanner$ has the quantum query complexity 
\[
\Omega\!\left(n^{1+\frac{3}{2k}+O(1/k^2)}\right).
\]
\end{corollary}
In the following, we first introduce the properties of $k$-spanner in \Cref{subsec:k-spanner} and analyze the properties of random skeleton graphs in \Cref{subsec:random-graph-construction}, which will be applied to the reduction from the BHVF problem to the $\spanner$ in \Cref{subsec:reduction-hbvf-spanner}. Together with \Cref{thm:bhvf}, we prove \Cref{thm:generic-lower-bound}.

\subsection{Introduction to k-Spanner}\label{subsec:k-spanner}
\begin{definition}[Multiplicative $k$-Spanner]
Given some integer $k \geq 1$ and a graph $G=(V,E)$, we say a subgraph $H=(V,E')$ of $G$ is a $k$-spanner of $G$ if for any two nodes $u,v \in V$ with $\dist_G(u,v)<\infty$,
\[
\dist_G(u,v)
\leq
\dist_H(u,v)
\leq
k \cdot \dist_G(u,v).
\]
\end{definition}

\begin{observation}\label{lem:mandatory-edge-criterion}
Given some integer $k \geq 1$ and a graph $G=(V,E)$, let $e=(u,v)\in E$ be an arbitrary edge in $G$. The following statements are equivalent:
\begin{enumerate}[label=(\roman*)]
\item Every $k$-spanner $H$ of $G$ contains the edge $e$.
\item $\dist_{G\setminus e}(u,v)\geq k+1$ where $G\setminus e=(V,E\setminus \{e\})$ denotes the graph obtained after deleting the edge $e$ from $E$.
\item Any cycle in $G$ that contains $e$ has length at least $k+2$.
\end{enumerate}
\end{observation}

\begin{proof}[Proof of \Cref{lem:mandatory-edge-criterion}]
Conditions (ii) and (iii) are equivalent by adjoining or deleting the edge $e$. Suppose condition (ii) holds and there exists a $k$-spanner $H=(V,E')$ of $G$ such that $e\notin E'$. Note that $H$ is a subgraph of $G\setminus e$, then
\[
\dist_H(u,v)\geq \dist_{G\setminus e} (u,v)\geq k+1=(k+1)\cdot \dist_G(u,v)
\]
which violates the $k$-spanner property. Conversely, suppose condition (i) holds and $\dist_{G\setminus e}(x,y)\leq k$. For every shortest path in $G$ with length $l$ that contains $e$,
we can construct a shortest path in $G\setminus e$ by replacing $e$ with the replacement path of $e$. Note that it has length $l-1+k \leq k \cdot l$ in $G\setminus e$. Therefore, the subgraph $G\setminus e$ is a $k$-spanner of $G$, which violates property (i).
\end{proof}
\begin{definition}\label{def:mismatch-budget}
For every integer $k\ge 7$, let
\begin{equation}\label{eqn:mu-defn}
\mu_k=\max\left\{
 j:
 \begin{array}{l}
 \exists \text{ odd integer } h \text{ such that} \\[-1mm]
 3\leq h\leq k, 
 h+ 2j\le k, j\le h+1
 \end{array}
\right\}.
\end{equation}
\end{definition}

\begin{observation}\label{lem:mu-closed-form}
For every $k\ge 7$, 
\[
\mu_k = \left\lfloor\frac{k+1}{3}\right\rfloor
-\mathbbm{1}[k\equiv 2\pmod 6] = 
\begin{cases}
    2q & k \in \{6q+1, 6q+2: q \geq 1\}\\
    2q+1 & k \in \{6q+3, 6q+4: q \geq 1\}\\
    2q+2 & k \in \{6q+5, 6q+6: q \geq 1\}
\end{cases}
\]
In particular, $\mu_k=k/3+O(1)$.
\end{observation}


\subsection{Random Skeleton Construction}\label{subsec:random-graph-construction}
In this section, we construct two random process and prove the existence of skeleton graphs with properties stated in \Cref{lem:exist-good-skeleton-graph}.

\paragraph{Random Skeleton Construction}
Fix an arbitrary $k\ge 7$ such that $\mu_k\geq 2$. Set $\mu=\mu_k$.
Consider any sufficiently large integer $N$. Let $S=N^{\mu-1}\geq N$. Set up two node sets $U$ and $W$ of size $|U|=|W|=S$. We construct a random skeleton graph $\Sigma=(U,W,E_\Sigma)$ such that for every pair of nodes $(u,w)$ with $u\in U, w\in W$, add the edge $(u,w)\in E_\Sigma$ independently with probability $p=\frac{cN}{S}$ where $c=0.01k^{-1}$. Set $B=B(\Sigma)=|E_\Sigma|$. Note that $\mathbb{E}_\Sigma[B]=S^2 \cdot p = cNS=cN^{\mu}$. Since every edge is sampled independently at random, by the Chernoff bound,
\begin{equation}\label{eq:chernoff-B}
    \Pr_{\Sigma}\left[\abs{B-c SN}\geq \frac{cSN}{2}\right] \leq 2e^{-\frac{1}{12} cSN}\leq 0.01.
\end{equation}
Label the edges in $E_\Sigma = \{\sigma_0, \sigma_1, \ldots, \sigma_{B-1}\}$.
For every edge $(u,w)\in E_\Sigma$ with $u\in U, w\in W$, we make it public such that its index $t\in [E_\Sigma]$ with $\sigma_t =(u,w)$ is public. We call $\Sigma$ the \textit{public skeleton graph}. We call nodes in $v\in U\cup W$ the \textit{skeleton nodes} and edges in $E_\Sigma$ the \textit{skeleton edges}.

\paragraph{Random Graph Construction}
Let $\Sigma=(U,W,E_\Sigma)$ be an arbitrary public skeleton graph generated by the above process. Let $B=|\Sigma_E|$. For every skeleton edge $\sigma_t$ with $t\in [B]$, sample a pair $(a_t,b_t)\in [N]\times [N]$ independently and uniformly. We call $(a_t,b_t)$ the hidden value for the skeleton edge $\sigma_t$.
Set up the string $r=(a_t,b_t)_{t\in[B]}$. In the following, we construct the graph $G_r =(V_r, E_r)$ described by $\Sigma$ and $r$.
\begin{enumerate}
    \item For every skeleton node $v\in U \cup W$, set up a \textit{leader} node $l_v$ and a set of \textit{coordinate}  node $C_v=\{c_{v, x}: x\in [N]\}$. We call $C_v$ the cluster for the skeleton node $v$.
    \begin{itemize}
        \item Add the edges $(l_v, c_{v, x})$ for every $c_{v,x}\in C_v$.
    \end{itemize}
    \item For every skeleton edge $\sigma_t = (u,w)\in E_\Sigma$ with $t\in [B]$, set up two gadget nodes $g_{t, u}, g_{t, w}$ and two dummy nodes $d_{t, u}, d_{t, w}$. Let $(a_t, b_t)\in [N]\times [N]$ be the hidden value for $\sigma_t$.
    \begin{itemize}
        \item Add the edge $e_t = (c_{u, a_t}, c_{w, b_t})\in E_r$ to $G_r$. 
        \item Add the edge $(c_{u,x}, g_{t, u})\in E_r$ to $G_r$ for every $c_{u,x}\in C_u$ with $x \neq a_t$.
        \item Add the edge $(c_{w,y}, g_{t, w})\in E_r$ to $G_r$ for every $c_{w,y}\in C_w$ with $y \neq b_t$.
        \item Add the edges $(g_{t, u}, d_{t, u}), (g_{t, w}, d_{t, w})\in E_r$ to $G_r$.
    \end{itemize}
\end{enumerate}
\begin{figure}
    \centering
    \includegraphics[width=0.7\linewidth]{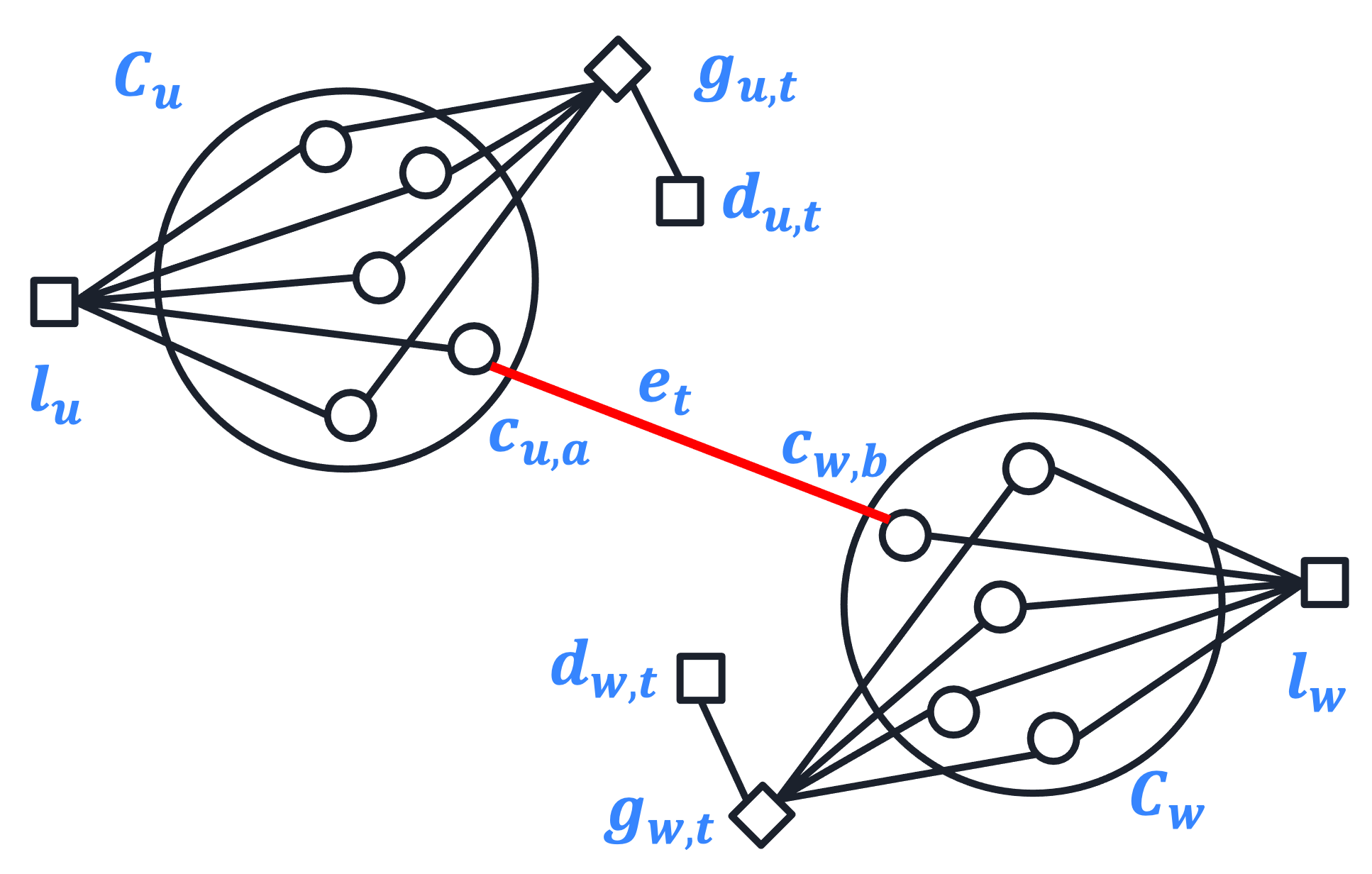}
    \caption{
    Consider the skeleton edge $\sigma_t =(u,w)\in E_\Sigma$ with $u\in U$ and $w\in W$ for some $t\in [B]$. Since the hidden value for block $t$ is $(a,b)$, the hidden edge for block $t$, highlighted in red, is $(c_{u,a}, c_{w,b})\in E_r$. For every unselected node $c_{u,x}\in C_u$ with $x\neq a$, add an incident edge to the gadget $g_{u,t}$. For every unselected node $c_{w,y}\in C_w$ with $y\neq b$, add an incident edge to the gadget $g_{w,t}$.
    }
    \label{fig:2 reduction to spanner}
\end{figure}
See \Cref{fig:2 reduction to spanner} for an illustration.
For every $t\in [B]$, we call $e_t\in E_r$ the hidden edge corresponding to $\sigma_t$.
Let $E_r^*=\{e_t: t\in [B]\}$ be set of hidden edges. We say two distinct skeleton edges $\sigma_i=(u_i,w_i),\sigma_j=(u_j,w_j)\in E_\Sigma$ are adjacent if they share one node, for example, either $u_i=u_j$ or $w_i=w_j$.

Set up the leader set $V_{l} = \{l_v\}_{v\in U \cup W}$, the coordinate set $V_c=\bigcup_{v\in U \cup W} C_v$, the gadget set $V_g = \{g_{t,u},g_{t,w}: \forall (u,w)\in E_\Sigma\}$ and the dummy set $V_d = \{d_{t,u},d_{t,w}: \forall (u,w)\in E_\Sigma\}$. Note that $V_r = V_l \sqcup V_c \sqcup V_g \sqcup V_d$ and it has size
\begin{equation}\label{eq:G-r-node-size}
    n=|V_r|= 2S+2SN + 4B.
\end{equation}
In addition, the edge set $E_r$ has size 
\begin{align}\label{eq:G-r-edge-size}
|E_r| 
& = \frac{1}{2}\left(|V_l|\cdot N + |V_g|\cdot N +  |V_d|\cdot 1 + 2S N+ \sum_{v\in U \cup W} (N \cdot \deg_\Sigma (v))
    \right)= 2SN+B+2BN.
\end{align}
\begin{observation}
    For every $\Sigma=(U,W,E_\Sigma)$, with $|E_\Sigma|\geq 1$, every random graph $G_r$ generated from $\Sigma$ is bipartite with girth $4$.
\end{observation}

\paragraph{Matching Events at Clusters} Fix a skeleton graph $\Sigma=(U,E,E_\Sigma)$ and fix a graph $G_r$ constructed from $\Sigma$ and $r$. For any skeleton node $v\in U \cup W$ and any two incident skeleton edges $(v,p), (v,q)\in E_\Sigma$, we say the corresponding hidden edges $(c_{v,i},c_{p,x}), (c_{v,j},c_{q,y})\in E_r$ are \textit{matching} at the cluster $C_v$ if they share the same node $c_{v,i}=c_{v,j}$. Otherwise, we say the hidden edges $(c_{v,i},c_{p,x}), (c_{v,j},c_{q,y})\in E_r$ are \textit{mismatching} at $C_v$.

\begin{observation}
    Given two adjacent skeleton edges $\sigma_i=(u_i,w'),\sigma_j=(u_j,w')\in E_\Sigma$, 
    \begin{equation}\label{eq:prob-matching-mismatching}
            \Pr_r [e_i,  e_j \text{ are matching}] = \frac{1}{N} \quad \text{ and }\quad \Pr_r [e_i,  e_j \text{ are mismatching}] = 1-\frac{1}{N}.
    \end{equation}
\end{observation}

\paragraph{Essential Hidden Edges}
Consider a random public skeleton graph $\Sigma=(U,W,E_\Sigma)$ and a random graph $G_r=(V_r, E_r)$ generated in the above process.
\begin{definition}[Essential Hidden Edges] \label{def:bad-event-Z}
    Go through every pair of skeleton nodes $(u,w)$ with $u\in U$ and $w\in W$. Given a hidden edge $e=(c_{u,a}, c_{w,b})$ for some $a,b\in [N]$, we say $e$ is an essential hidden edge if, in the subgraph $G'_r=G_r\setminus e$, the distance between $c_{u,a}$ and $c_{w,b}$ is $\dist_{G'_r}(c_{u,a}, c_{w,b})\geq k+1$. On the other hand, we say $e$ is a nonessential hidden edge if $\dist_{G'_r}(c_{u,a}, c_{w,b})\leq k$.
 
    Set up the following indicator random variables over the randomness of $\Sigma$ and $r$ such that
    \[
    Z_{u,w}=\mathbbm{1}[\text{there exists a \underline{nonessential} hidden edge } (c_{u,a}, c_{w,b})\text{ for some }a,b\in [N] ]
    \]
    \[
    \overline{Z}_{u,w}=\mathbbm{1}[\text{there exists an \underline{essential} hidden edge } (c_{u,a}, c_{w,b})\text{ for some }a,b\in [N] ]
    \]
    Set up $Z=\sum_{u\in U, w\in W} Z_{u,w}$ and $\overline{Z}=\sum_{u\in U, w\in W} \overline{Z}_{u,w}$ to be the number of nonessential hidden edges and essential hidden edges in $G_r$, respectively. 
\end{definition}
Note that the number of hidden edges in $G_r$ is $B=Z + \overline{Z}$. By \Cref{lem:mandatory-edge-criterion}, every $k$-spanner of $G_r$ contains all the essential hidden edges. In the following, we compute the number of nonessential hidden edges, denoted by $Z$, over the randomness of the skeleton graph $\Sigma$ and the hidden values $r$. We start from the following observations. \begin{figure}
    \centering
    \includegraphics[width=0.5\linewidth]{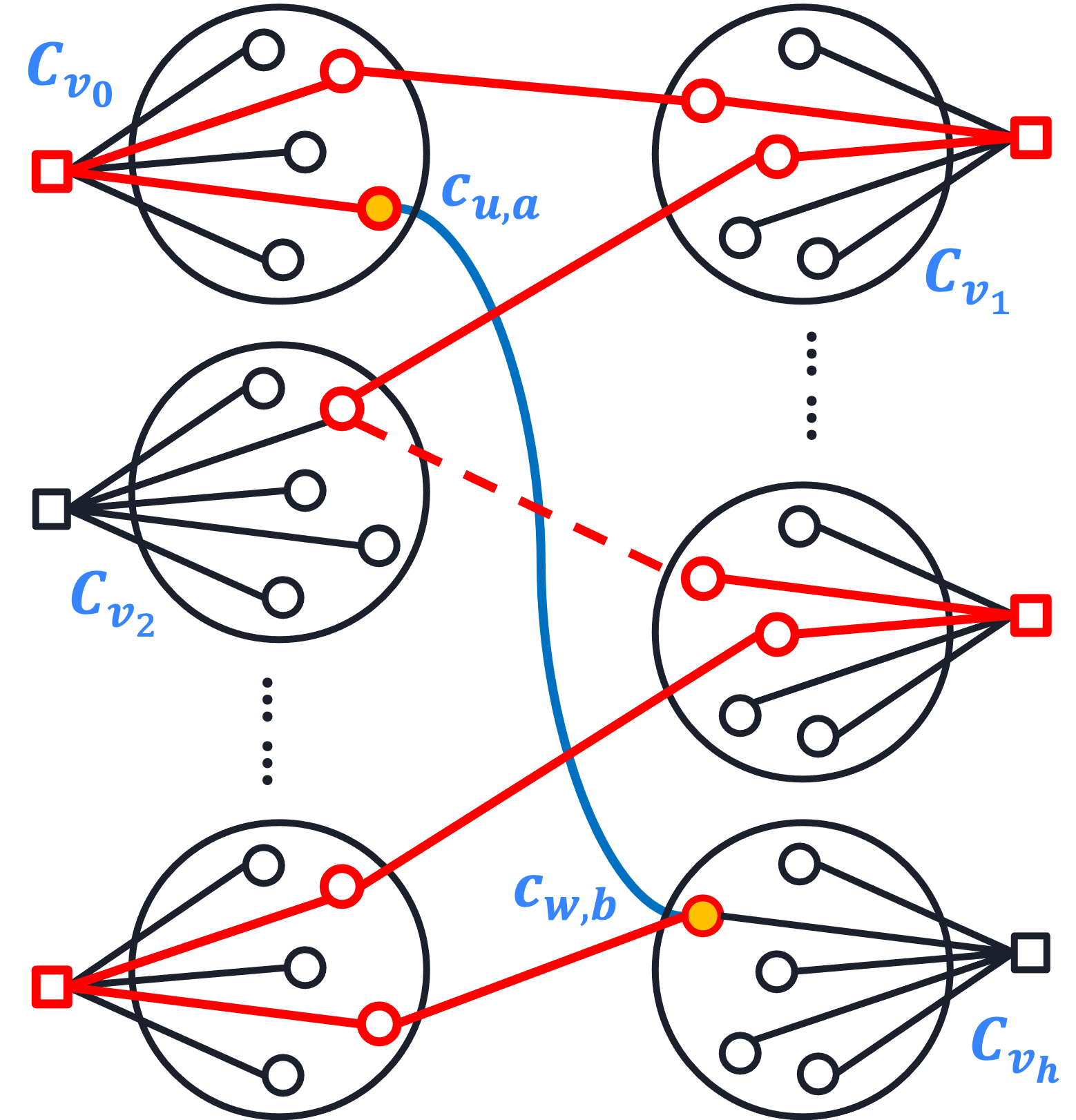}
    \caption{
    Fix two arbitrary coordinate nodes $c_{u,a}$ and $c_{w,b}$ such that $u\in U, w\in W$ for some $a,b\in [N]$. Assume the hidden edge $e=(c_{u,a}, c_{w,b})\in E_r$ exists (highlighted in blue). 
    Fix an arbitrary shortest path $P$ in $G_r\setminus e$ between $c_{u,a}$ and $c_{w,b}$ (highlighted in red). Let $S_P=(v_0=u, v_1, \ldots, v_{h}=w)$ be the projected skeleton sequence of $P$. For the junctions $v_2$ and $v_h$, they belong to the first case in \Cref{obs:shortest-path-candidate}. For the junctions $v_0, v_1, v_{h-2}$ and $v_{h-1}$, they belong to the second case.
    }
    \label{fig:3 reduction-shortest-path}
\end{figure}
\begin{observation}\label{obs:shortest-path-candidate}
    Assume the skeleton edge $(u,w)\in E_\Sigma$ exists. Let $e=(c_{u,a},c_{w,b})\in E_r$ be the corresponding hidden edge. Consider the graph $G'_r=G_r \setminus e$ by removing the edge $e$ from $E_r$. Assume the two coordinate nodes $c_{u,a},c_{w,b}$ are connected in $G'_r$. Fix an arbitrary shortest path $P$ in $G'_r$ between $c_{u,a}$ and $c_{w,b}$ with length $l=|P|$. It satisfies the following properties:
    \begin{itemize}
        \item No dummy nodes occur on the path $P$.
        \item The path length $l$ is odd and $l\geq 3$.
        \item The number of hidden edges on the path $P$, denoted by $h$, is odd and $h\geq 3$. 
        \item Go through every edge of the path $P$ sequentially starting from $c_{u,a}$, we say the path $P$ crosses clusters if the edge $e=(c_{v, x}, c_{v', y})$ is a hidden edge between $C_{v}$ and $C_{v'}$. Given $h$ hidden edges, let
        \[
        S_P =(v_0, v_1, \cdots, v_{h})
        \]
        be the sequence of skeleton nodes such that $(C_{v_0}, C_{v_1}, \cdots, C_{v_h})$ is the sequence of clusters that the path $P$ goes across. Note that $v_0=u$ and $v_h=w$. Note that for every cluster $C_v$, any two coordinate nodes $c_{v,x}$ and $c_{v,y}$ have distance at most $2$ via the leader $l_v$. Then all the skeleton nodes $v_i$ in the sequence $S_P$ are distinct. Otherwise, $P$ is not a shortest path. 
        
        We call $S_P$ the projected skeleton sequence of $P$ and every skeleton node $v_i$ in $S_P$ a junction.
        \item For every junction $v_i$, when the path $P$ enters the cluster $C_{v_i}$ via a hidden edge, it has the following two choices:
        \begin{enumerate}
            \item enters the cluster $C_{v_{i+1}}$ via the next hidden edge,
            \item takes a detour to the leader $l_{v_i}$ or some gadget $g_{v_i, t}$, returns to the cluster $C_{v_{i}}$ and then enters the cluster $C_{v_{i+1}}$ via the next hidden edge. 
        \end{enumerate}
        The first case happens when the corresponding hidden edges of $(v_{i-1},v_{i}),(v_{i},v_{i+1})\in E_\Sigma$ are matching at $C_{v_i}$. The second case happens when the corresponding hidden edges of $(v_{i-1},v_{i}),(v_{i},v_{i+1})\in E_\Sigma$ are mismatching at $C_{v_i}$. The detour has length $2$.
        \item Let $j\leq h+1$ be the number of mismatching events on the path $P$. Therefore
        \[
        l = |P|=h+2j.
        \]
    \end{itemize}
\end{observation}
See \Cref{fig:3 reduction-shortest-path} for an illustration. 
\begin{lemma}\label{lem:upperbound-Z-u-w}
    Fix any two nodes $u\in U$ and $w\in W$. There exists a nonessential hidden edge between $C_u$ and $C_w$ with probability 
    \[
    \lambda = \mathbb{E}_{\Sigma, r}[Z_{u,w}] \leq \frac{p}{200}.
    \]
\end{lemma}
\begin{proof}[Proof of \Cref{lem:upperbound-Z-u-w}]
    Assume there exists a hidden edge $e=(c_{u,a}, c_{w,b})\in E_r$ for some $a,b\in [N]$. Note that it occurs with probability $p$ independently. 
    Fix the nodes $c_{u,a}\in C_u$ and $c_{w,b}\in C_w$. 
    In the following, we compute the probability that there exists a shortest path $P$ in $G'_r=G_r \setminus e$ between $c_{u,a}$ and $c_{w,b}$ with length  $l=|P|\leq k$.

    First, we say an integer pair $(h,j)$ is $k$-valid if
    (1) $h$ is odd and $3\leq h\leq k$ (2) $0\leq j\leq h+1$ and (3) $h+2j \leq k$. By \Cref{def:mismatch-budget}, 
    \[
    \mu_k=\max\left\{
    j:
    \begin{array}{l}
    \exists \text{ odd integer } h \text{ such that} \\[-1mm]
    3\leq h\leq k, 
    h+ 2j\le k, j\le h+1
    \end{array}
    \right\}.
    \]
    every $k$-valid pair $(h,j)$ has $j\leq \mu_k$.
    Consider every possible shortest-path candidate $P$ in the random graph $G'_r$ between $c_{u,a}$ and $c_{w,b}$ with length at most $k$. We use $h$ to denote the number of hidden edges in $P$ and $j$ to denote the number of mismatching events on $P$. By \Cref{obs:shortest-path-candidate}, the pair $(h, j)$ is $k$-valid.
    
    Fix a $k$-valid pair $(h,j)$. Consider every possible candidate $P$ with $h$ hidden edges and $j$ mismatching events. By \Cref{obs:shortest-path-candidate}, let
    \[
    S_P = (v_0=u, v_1,\ldots, v_h=w)
    \]
    be the projected skeleton sequence of $P$ where every skeleton node $v_i$ in $S_P$ is distinct. Since $v_0=u$ and $w_h=w$ are fixed, there are at most $S^{h-1}$ possible sequences where $S=|U|=|W|=N^{\mu-1}$. In addition, among $h+1$ junctions, there are $j$ mismatching events and $h+1-j$ matching events. Then, there are $\binom{h+1}{j}$ possible arrangements.

    Fix one possible sequence and one possible arrangement. For every $(v_i, v_{i+1})$ with $i\in [h]$, the hidden edge across $C_{v_i}$ and $C_{v_{i+1}}$ exists if and only if the corresponding skeleton edge $(v_i, v_{i+1})\in E_\Sigma$ exists. By construction, the skeleton edge exists independently with probability $p$. 
    Among $h+1$ junctions, there are $j$ mismatching events, and each occurs independently with probability $1-1/N$. In addition, there are $h+1-j$ matching events, and each occurs independently with probability $1/N$. 
    Therefore, such a shortest-path candidate exists with probability
    \[
    p^{h} \cdot \left(1-\frac{1}{N}\right)^{j} \cdot \left( \frac{1}{N}\right)^{h+1-j}
    \leq
    p^{h} \cdot \frac{N^{j}}{N^{h+1}}
    \leq
    p^{h} \cdot \frac{N^{\mu}}{N^{h+1}}
    \]
    where $j \leq \mu=\mu_k$ by \Cref{def:mismatch-budget}. 
By the union bound, a shortest-path candidate with $h$ hidden edges and $j$ mismatching events exists with probability at most
    \begin{align*} 
    S^{h-1}\cdot \binom{h+1}{j} \cdot p^{h} \cdot \frac{N^{\mu}}{N^{h+1}} 
    & \leq  
    S^{h-1}\cdot 2^{h+1} \cdot \left(\frac{cN}{S}\right)^{h}\cdot \frac{N^{\mu}}{N^{h+1}} = 2^{h+1} c^h \cdot \frac{N^\mu}{SN} \\
    & = 2^{h+1} c^h.
    \end{align*}
Go through all $k$-valid pairs. By the union bound, given any two nodes $c_{u,a}$ and $c_{w,b}$ such that $(c_{u,a},c_{w,b})\in E_r$, 
    \begin{align*}
    \Pr[ \dist_{G'_r}(c_{u,a}, c_{w,b})\leq k] & =\Pr[ \exists \text{ a shortest-path candidate of length}\leq k ]\\
    & \leq \sum_{\substack{\text{odd }h: \\ 3\leq h\leq k}} \sum_{\substack{j: 0\leq j\leq k+1\\ 2j+h\leq k}} 2^{h+1}c^h \leq k^2 \cdot 2\cdot  (2c)^{3}\\
    & \leq \frac{1}{200}
    \end{align*}
    where $c=0.01k^{-1}$ and $h\geq 3$.
Since every skeleton edge is sampled independently with probability $p$, the overall probability is
    \begin{align*}
    \lambda  = \Pr_{\Sigma , r}[Z_{u,w}=1]
    & = \Pr[\exists a,b\in [N]:(c_{u,a}, c_{w,b})\in E_r] \cdot \Pr[ \dist_{G'_r}(c_{u,a}, c_{w,b})\leq k |(c_{u,a}, c_{w,b})\in E_r]
    \\
    & \leq p \cdot \frac{1}{200}.
    \end{align*}
\end{proof}
\begin{lemma}\label{lem:exist-good-skeleton-graph} 
    For every integer $k\geq 7$ and for all sufficiently large $N$, there exists a deterministic public skeleton graph $\Sigma^*_N=(U,W, E^*_\Sigma)$ with $S=|U|=|W|=N^{\mu_k-1}$ and $B=|E_\Sigma^*|$ such that
    \[
    \frac{cS N}{2}   \leq
    B \leq \frac{3 cS N}{2} \quad \text{ and }\quad     \Pr_r\left[Z \geq \frac{2 cSN}{5}\middle|\Sigma_N^*\right] \leq \frac{1}{40}
    \]
    where $c=0.01k^{-1}$. 
\end{lemma}
\begin{proof}[Proof of \Cref{lem:exist-good-skeleton-graph}]
    Let $\mathcal{E}$ be the event that $\frac{cS N}{2} \leq
    B(\Sigma) \leq \frac{3 cS N}{2}$. By \Cref{eq:chernoff-B}, $\Pr_{\Sigma}[\mathcal{E}]\geq 0.99$.
    Given $Z=\sum_{u\in U, w\in W} Z_{u,w}$, then 
    \[
    \mathbb{E}_{\Sigma, r}[Z]\leq S^2 \cdot \frac{p}{200} = \frac{cSN}{200}.
    \]
    For every possible realization $\Sigma$, we define $f(\Sigma)=\mathbb{E}_r[Z|\Sigma]$
    as the expectation of $Z$ conditionally on the skeleton graph $\Sigma$.
    By the law of total expectation,
    \[
    \mathbb{E}_{\Sigma,r}[Z]=\mathbb{E}_{\Sigma}[\mathbb{E}_r[Z|\Sigma]]=\mathbb{E}_{\Sigma}[f(\Sigma)]
    \]
    \[
    \mathbb{E}_{\Sigma}[f(\Sigma)] = \Pr[\mathcal{E}]\cdot 
    \mathbb{E}_{\Sigma}[f(\Sigma)|\mathcal{E}] + \Pr[\overline{\mathcal{E}}]\cdot 
    \mathbb{E}_{\Sigma}[f(\Sigma)|\overline{\mathcal{E}}].
    \]
    Since $Z\geq 0$ is nonnegative, then $\mathbb{E}_{\Sigma}[f(\Sigma)|\mathcal{E}] \leq 
    \frac{\mathbb{E}_{\Sigma}[f(\Sigma)]}{\Pr[\mathcal{E}]} \leq 2 \cdot \mathbb{E}_{\Sigma, r}[Z] \leq \frac{cS N}{100}$.
    Therefore, there exists a deterministic skeleton graph $\Sigma_N^*$ such that 
    \[
    \frac{cS N}{2} \leq
    B(\Sigma_N^*) \leq \frac{3 cS N}{2} \quad \text{ and } \quad f(\Sigma_N^*)=\mathbb{E}_{r}[Z|\Sigma_N^*] \leq \frac{cS N}{100}.
    \]
    By Markov's Inequality, the skeleton graph $\Sigma_N^*$ has $Z\geq cSN/5$ with probability
    \[
    \Pr_r \left[Z \geq \frac{2cSN}{5}\middle|\Sigma_N^*\right] \leq \frac{\mathbb{E}_r [Z|\Sigma_N^*]}{2cSN/5}\leq 1/40.
    \] 
\end{proof}

\subsection{Reduction from BHVF to k-Spanner}\label{subsec:reduction-hbvf-spanner}
Given any integer $k\geq 7$, consider any sufficiently large $N$. 
We regard $k$ and $\mu=\mu_k$ as constants. 
Fix the graph $\Sigma^*_N=(U,W,E_\Sigma^*)$ that satisfies the properties in \Cref{lem:exist-good-skeleton-graph}. Set $B=|E_\Sigma^*|=\Theta(cSN)$.
\paragraph{Graph Construction and Properties} Consider any input $r=(r_t^A, t_t^B)_{t\in[B]}=(a_t,b_t)_{t\in [B]}$ of BHVF problem where every $a_t,b_t \in [N]$. We construct a graph $G_r=(V_r, E_r)$ from the skeleton graph $\Sigma^*_N$ and the string $r$ by the random process in \Cref{subsec:random-graph-construction}. By \Cref{eq:G-r-node-size} and \Cref{eq:G-r-edge-size}
\[
n=|V_r|=\Theta(B)=\Theta(cN^{\mu})
\quad \text{ and } \quad
|E_r|=\Theta(BN)=\Theta(cN^{\mu+1})=\Theta(n^{1+\frac{1}{\mu}}).
\]
Let  $\mathcal{G}_n =\{G_r\}$ be the graph class for every possible $r=(a_t, b_t)_{t\in[B]}$ with $a_t, b_t\in [N]$. 
We claim the following two properties of $G_r$ generated by the fixed skeleton $\Sigma^*_N$ and the random input $r$.

\begin{lemma}\label{lem:Gr-essential-edges-linear}
    With probability at least $39/40$, the graph $G_r=(V_r,E_r)$ contains at least $B/5$ essential hidden edges, which implies that every $k$-spanner of $G_r$ contains at least $\Omega(n)$ edges.
\end{lemma}

\begin{lemma} 
\label{obs:k-spanner-Gr}
    For every graph $G_r=(V_r, E_r)\in \mathcal{G}_n$, it is bipartite with girth $4$. In addition, there always exists a $3$-spanner $H_r=(V_r, E'_r)$ of $G_r$ with size $|E'_r|=O(n)$.
\end{lemma}
\begin{proof}[Proof of \Cref{lem:Gr-essential-edges-linear}]
    Recall that $Z$ and $\overline{Z}$ denote the number of nonessential hidden edges and essential hidden edges in $G_r$, respectively. By \Cref{lem:exist-good-skeleton-graph}, 
    \[
    \Pr_r\left[Z \geq \frac{2 cSN}{5}\middle|\Sigma_N^*\right] \leq \frac{1}{40}.
    \]
    Given the skeleton $\Sigma_N^*$, the total number of hidden edges is $Z+\overline{Z}=B\geq cSN/2$.
    Conditionally on $Z<2cSN/5$, the number of essential hidden edges is
    \[
    \overline{Z}=B \cdot \left(1-\frac{Z}{B} \right)\geq B \cdot \left(1-\frac{2cSN/5}{cSN/2} \right)= \frac{1}{5}B.
    \]
\end{proof}

\begin{proof}[Proof of \Cref{obs:k-spanner-Gr}]
    It is easy to see that every $G_r=(V_r, E_r)\in \mathcal{G}_n$ is bipartite with girth $4$. In the following, we construct a $3$-spanner $H_r=(V_r, E'_r)$ of $G_r$ by deleting edges in $E_r$.
    
    Go through every gadget node $g_{v, t}$ for $v \in U \cup W$ and $t\in [B]$. Note that it is adjacent to $N-1$ coordinate nodes $c_{v,x}$ for some $x\in [N]$. Keep an arbitrary edge $(g_{v,t}, c_{v,x_v})$ among these $N-1$ edges and remove all the remaining $N-2$ edges. Let $E^{\text{rm}}_r$ be the set of all removed edges. Set $E'_r=E_r \setminus E_r^{\text{rm}}$. Note that $|E_r^{\text{rm}}|=2B(N-2)$ and
    \[
    |E_r'|=|E_r|-|E_r^{\text{rm}}|=2SN+B+2BN-(2BN-4B) =2SN+5B=O(n).
    \]
    For every edge $(g_{v,t}, c_{v,x})\in E_r^{\text{rm}}$, let $c_{v,x_v}\in C_v$ be the coordinate node such that $(g_{v,t}, c_{v,v_x})\in E_r'$. Then, 
    there exists a path $(g_{v,t}, c_{v,v_x}, l_v, g_{v, x})$ between $g_{v,t}$ and $c_{v,x}$ in $H_r=(V_r, E_r')$ with length $3$. Therefore $H_r=(V_r, E_r')$ is a $3$-spanner of $G_r$ with $O(n)$ edges.
\end{proof}

\paragraph{Simulation with BHVF Oracles} Now we prove the following main result.
\begin{theorem}\label{thm:Reduction-BHVF-to-spanner}
    If there exists an algorithm $\mathcal{B}$ (classical or quantum) that solves $(n, k)$-$\spanner$ over $\mathcal{G}_n$ using $Q$ graph queries with probability at least $1-\delta$, there exists an algorithm $\mathcal{A} $ (respectively classical or quantum) that solves $(B=\Theta(n),N=\Theta(n^{1/\mu_k}),d=\frac{1}{5})$-$\bhvf$ using $O(Q)$ $\bhvf$ oracle queries with probability at least $39/40-\delta$.
\end{theorem}

\begin{proof}[Proof of \Cref{thm:Reduction-BHVF-to-spanner}]
Let $\mathcal{B}$ be the given algorithm that solves $(n,k)$-$\spanner$ over $\mathcal{G}_n$ using $Q$ queries. Firstly, note that the algorithm $\mathcal{B}$ can access to adjacency queries, neighborhood queries and degree queries to the input graph. Let all the degree queries be free. In the following, we prove that every adjacency query and every neighborhood query over $\mathcal{G}_n$ can be simulated with $O(1)$ calls to the $\bhvf$ oracle.
\begin{description}
    \item[Adjacency Queries] Given any two input nodes $p,q\in V_r$, we first categorize $p,q$ in the categories of the leader nodes, the coordinate nodes, the gadget nodes, and the dummy nodes. 
    Now we proceed via the following cases:
    \begin{enumerate}
        \item Assume (w.l.o.g.) $p=l_v\in V_l$ for some $v\in U \cup W$. Return $\mathcal{A}(p,q)=1$ if and only if $q=c_{v, i}$ for $i\in [N]$.
        \item Assume (w.l.o.g.) $p=d_{t, u}\in V_d$ for some $t\in [B]$ and $u\in U\cup W$. Return $\mathcal{A}(p,q)=1$ if and only if $q=g_{t, u}$. 
        \item Assume (w.l.o.g.) $p=g_{t, u}\in V_g$ for some $t\in [B]$ and $u\in U\cup W$. In addition, we assume $q=c_{u,x}\in C_u$ for some $x\in [N]$, since otherwise we can return $\mathcal{A}(p,q)=0$. If $u\in U$, set $\tau=A$. Otherwise $u\in W$ and set $\tau=B$.
        Query $O_r(\tau, t, x)$ and return $\mathcal{A}(p,q)=1$ if and only if $O_r(\tau, t, x)=0$. 
        \item In the last case, both $p,q$ are coordinate nodes. W.l.o.g, we assume $p=C_{u,x}$ for some $u\in U, x\in [N]$ and $q=C_{w,y}$ for some $w\in W, y\in [N]$. Otherwise, we can return $\mathcal{A}(p,q)=0$.
        If $(u,w)\notin E_\Sigma^*$, then $\mathcal{A}(p,q)=0$. Otherwise, let $t \in [B]$ be the index such that $\sigma_t=(u,w)\in E_\Sigma^*$. Query $O_r(A,t,x)$ and return $\mathcal{A}(p,q)=1$ if and only if $O_r(A,t,x)=y+1$.
    \end{enumerate}
    \item[Neighborhood Queries] Given an input node $p\in V_r$ and an integer $x \in \mathbb{N}_{\geq 0}$, we assume $l \in [\deg_{G_r}(p)]$ since $\deg_{G_r}(p)$ is given for free.
    \begin{itemize}
        \item For the adjacency list of a leader node $p=l_v\in V_l$, it has length $N$. Given any $x \in [N]$, return $\mathcal{N}(p,x) = c_{v,x}$.
        \item For the adjacency list of a dummy node $p=d_{t,v}\in V_d$, it has length $1$ and return $\mathcal{N}(p,0) = g_{t,v}$.
        \item For the adjacency list of a gadget node $p=g_{t,v}\in V_g$ with some $v\in U \cup W$, it has length $N$. If $v\in U$, set $\tau=A$. Otherwise $v\in W$ and set $\tau=B$.
        Given any $x \in [N]$, query $O_r(\tau, t, x)$ and return 
         \[
        \mathcal{N}(p,x)=\begin{cases}
            d_{t,v} \quad & \text{if } O_r(\tau, t,x)\neq 0,  \\
            c_{v,x} & \text{if } O_r(\tau,t,x)=0.
        \end{cases}
        \]
        \item For the adjacency list of a coordinate node $p=c_{v,y}\in C_v$ with some $v\in U \cup W$, it has length $1+\deg_{\Sigma^*_N}(v)$. 
        Given $x=\deg_{\Sigma^*_N}(v)$, return $\mathcal{N}(p,x)=l_v$.
        
        If $v\in U$, set $\tau=A$. Otherwise $v\in W$ and set $\tau=B$.
        Given any $x \in [\deg_{\Sigma^*_N}(v)]$, let $v_x\in U\cup W$ be the $x$-th neighbor of $v$ in $\Sigma^*_N$ and let $t_x\in [B]$ be the index such that $\sigma_{t_x}=(v, v_x)\in E_\Sigma^*$. Query $O_r(\tau, t_x, y)$ and return 
         \[
        \mathcal{N}(p,x)=\begin{cases}
            c_{v_x,z-1} \quad & \text{if } O_r(\tau, t_x, y)= z \text{ for some } z\in [1,N+1],  \\
            g_{v,t_x} & \text{if } O_r(\tau, t_x, y)=0.
        \end{cases}
        \]
    \end{itemize}
    \end{description}
    Therefore, every adjacency query and every neighborhood query over $\mathcal{G}_n$ can be simulated with $O(1)$ calls to the $\bhvf$ oracle.

    Given a random input $r=(r_t)_{t\in [B]}$, we construct the corresponding graph $G_r$ from the fixed skeleton $\Sigma_N^*$. Let $\overline{Z}$ be the number of essential hidden edges. Set up $\mathcal{F}$ to be the event that $\overline{Z}\geq B/5$. By \Cref{lem:Gr-essential-edges-linear}, 
    $\Pr_r[\mathcal{F}]\geq 39/40$.
    We can simulate the algorithm $\mathcal{B}$ that solves $(n,k)$-$\spanner$ with probability $1-\delta$ with quantum query complexity $O(Q)$. Let $H$ be the output of $\mathcal{B}$.
    Assume the event $\mathcal{F}$ happens and the output $H$ is a $k$-spanner. It occurs with probability
    \begin{align*}
        \Pr_{r, \mathcal{B}}[\mathcal{F} \text{ and } H \text{ is a }k\text{-spanner}] 
    & \geq 
    1-\Pr_{r, \mathcal{B}}[\overline{\mathcal{F}}]-\Pr_{r, \mathcal{B}}[ H \text{ is a not }k\text{-spanner}]\\
    & \geq  1-\frac{1}{40}-\delta =  \frac{39}{40}-\delta.
    \end{align*}
    Since every $k$-spanner contains all the essential hidden edges, the output $H$ contains all the essential hidden edges. For every hidden edge $e_t=(c_{u,a}, c_{w,b})$ in the spanner $H$ , recover the hidden value of block $t$ by setting $(r_t^A, r_t^B)=(a, b)$.
    
    Note that there are at least $B/5$ essential hidden edges in $G_r$. By utilizing the output of algorithm $\mathcal{B}$, we can recover at least $B/5$ hidden values.
    The overall success probability over both the inputs and the algorithm is at least $39/40-\delta$. 

\end{proof}
\paragraph{Proof of \Cref{thm:generic-lower-bound}} Now we complete the proof of \Cref{thm:generic-lower-bound}.
Given any integer $k\geq 7$, set $\mu=\mu_k\geq 2$.
For every sufficiently large $N$, set up the skeleton graph $\Sigma_N^*=(U,W,E_\Sigma^*)$. Set $B=|E_\Sigma^*|=\Theta(N^{\mu})$ and $n=2N^{\mu-1}+2N^{\mu}+4B=\Theta(N^{\mu})$. Set up the graph class $\mathcal{G}_n =\{G_r\}$. Consider any algorithm $\mathcal{B}$ that solves $(n,k)$-$\spanner$ over $\mathcal{G}_n$ with $Q$ quantum graph queries. By \Cref{thm:Reduction-BHVF-to-spanner}, the task $(B, N, d=\frac{1}{5})$-$\bhvf$ can be solved with $Q'=O(Q)$ quantum queries to the $\bhvf$ oracle. By \Cref{thm:bhvf}, $Q'=\Omega(B\sqrt{N})$. Therefore,
\[
Q=\Omega\left(B \sqrt{N}\right)
=\Omega\left(N^{\mu+{\frac{1}{2}}}\right)
=\Omega\left(n^{1+{\frac{1}{2\mu}}}\right).
\]

\section{Lower bounds for the BHVF Problem}\label{sec:BHVF-lb-proof}
In this section, we give the proof of \Cref{thm:bhvf}.
\subsection{The Standard Query Model}

\paragraph{The Input Space} For each position $j\in [N]=\{0,1,\ldots, N-1\}$, we use $\ket{j}$ to denote the quantum state corresponding to the column vector of size $N$ with value $1$ at position $j$ and $0$ everywhere else. For any $i,j\in [N]$, note that $\brakett{i}{j}=\mathbbm{1}[i=j]$. We use $\mathcal H_{\inp}$ to denote the input space over all $B$ blocks such that
\[
\mathcal H_{\inp}=\left(\mathcal H_{\inp,\loc} \mathcal\otimes H_{\inp,\loc} \right)^{\otimes B} \qquad \text{where }     \mathcal H_{\inp,\loc}=\spanop\{\ket{0}, \ket{1},\ldots, \ket{{N-1}}\}.
\]
A computational basis state of \(\mathcal H_{\inp}\) is written as
\[
    \bigotimes_{t\in [B]}\ket{i_t,j_t}
    \qquad \text{for any }
    i_t, j_t\in [N]=\{0,\ldots,{N-1}\}.
\]
\paragraph{The Algorithm Space} The algorithm space $\mathcal{H}_\Alg$ is over three registers $Q,P,W$, where the query register $Q$ holds a block index $t\in [B]$, a type index $\tau \in \{A,B\}$ and a position index $x\in [N]$, the phase register $P$ holds $u\in \mathbb{Z}_{N+1}$
and the working register $W$ holds some value $w$. The algorithm space \(\mathcal H_{\Alg}\) is spanned by basis states
\[
    \ket{t,\tau , x}_Q\ket{u}_P\ket{w}_W.
\]
We may drop the subscript $QPW$ when it is clear from the context. The workspace contains a designated classical output substring, which is measured only at the end.
\paragraph{The Query Oracle} Given an input $r=(r_t)_{t\in [B]}=(r^A_t, r^B_t)_{t\in [B]}$ with $r^A_t, r^B_t \in [N]$, the query oracle $\mathcal{O}_r$ is defined by
\[
\mathcal{O}_r \ket{t,\tau,x,u,w} =  \omega^{u \cdot f_t (\tau, x)} \ket{t,\tau,x,u,w}
\qquad \text{where } \omega=e^{2\pi i/(N+1)}.
\]
Recall that $f_t(A, x)= (r_t^B +1) \cdot \mathbb{1}[x=r_t^A]$ and $f_t(B, x)= (r_t^A +1) \cdot \mathbb{1}[x=r_t^B]$.
\paragraph{Quantum Query Algorithms in the Standard Model} A $T$-query algorithm is specified by a sequence of $U_0, U_1, \ldots, U_T$ of unitary operators acting on the algorithm space $\mathcal{H}_\Alg$. Given an input $r=(r_t)_{t\in [B]}$, the state $\ket{\psi_t^r}$ of the algorithm after $t\leq T$ queries to some input $r$ is 
\[
    \ket{\psi^r_t}
    =
    U_t\mathcal{O}_rU_{t-1}\cdots U_1\mathcal{O}_rU_0
    \ket0_{\Alg}.
\]
Next, we consider the quantum state over the algorithm space $\mathcal{H}_{\Alg}$ and the input space $\mathcal{H}_\inp$. The query operator $\mathcal{O}$ is a unitary transformation acting on the joint space $\mathcal{H}_{\Alg}\otimes \mathcal{H}_\inp$ such that
\[
\mathcal{O}\ket{t,\tau, x,u,w}\ket{r} =(\mathcal{O}_r\ket{t,\tau, x,u,w})\otimes\ket{r}.
\]
Note that the input $r=(r^A_t, r^B_t)_{t\in [B]}$ is sampled from $\mathcal{D}_{B,N}=\bigotimes_{t\in [B]}\text{Uniform}([N]\times [N])$.
The joint state $\ket{\psi_t}$ over the joint space $\mathcal{H}_{\Alg}\otimes \mathcal{H}_\inp$ after $t\leq T$ queries to some input $r$ is 
\[
    \ket{\psi_t}
    =
    \left(U_t\otimes \mathbb{I}_{\inp}\right) 
    \mathcal{O} 
    \left(U_{t-1}\otimes \mathbb{I}_{\inp}\right) 
    \ldots 
    \mathcal{O}\left(U_0\otimes \mathbb{I}_{\inp}\right) 
    \ket{0}_{\Alg}  \ket{\mathcal{D}}_{\inp}.
    \]
where the initial input state is $\ket{\mathcal{D}}_\inp= [\ket{\tilde{u}}\otimes \ket{\tilde{u} }]^{\otimes B}$ with $ \ket{\tilde{u}}=\frac1{\sqrt N}\sum_{i\in [N]}\ket{i}$.
\subsection{The Recording Query Model}
\paragraph{The Recording Space} First, we introduce an additional state $\ket{\perp}$ that is orthogonal to 
\\$\spanop\{\ket{0}, \ket{1},\ldots, \ket{{N-1}}\}$. We use $\mathcal H_{\rec}$ to denote the recording space over all $B$ blocks such that
\[
\mathcal H_{\rec}=\left(\mathcal H_{\rec,\loc}\otimes \mathcal H_{\rec,\loc} \right)^{\otimes B} \qquad \text{where }     \mathcal H_{\rec,\loc}=\spanop\{\ket{\perp}, \ket{0}, \ket{1},\ldots, \ket{{N-1}}\}.
\]
A computational basis state of \(\mathcal H_{\rec}\) is written as $ \bigotimes_{t\in [B]}\ket{i_t, j_t}$ for any $i_t, j_t\in\{\perp\}\cup [N]$. In this Hilbert space, we extend the Query Oracle as identity whenever $r_t$ has a $\perp$ coordinate.
\begin{definition}[Local Swap]\label{def:local-swap-S}
Set up the unitary operator \(S:\mathcal H_{\rec, \loc}\to\mathcal H_{\rec, \loc}\) such that 
\[
S:
\begin{cases}
    \ket\perp \rightarrow \ket{\tilde{u}}, \\
    \ket{\tilde{u}} \rightarrow \ket\perp,\\
    \ket\phi \rightarrow \ket\phi, \qquad \text{if } \ket{\phi} \in \spanop\{\ket{\perp}, \ket{\tilde{u}}\}^{\perp}
\end{cases}
\]
where $\ket{\tilde{u}}=\frac{1}{\sqrt{N}}\sum_{i\in [N]}\ket{i}$. Specifically, for each state $\ket{i}$ with $i\in [N]$, 
\[
S: 
\ket{i} \rightarrow
\ket{i}+
\frac{1}{\sqrt{N}}\ket{\perp} -\frac{1}{\sqrt{N}}\ket{\tilde{u}}.
\]
\end{definition}

\begin{lemma}[Matrix Representation of $S$] \label{lem:S-matrix}
Consider the basis of $\mathcal{H}_{\rec, \loc}$ in the order $\bigl(\ket\perp,\ket{0},\ket{1},\ldots,\ket{{N-1}}\bigr)$. Then, 
\begin{equation}
\label{eq:S-matrix}
    S
    =
    \begin{pmatrix}
        0 & \frac1{\sqrt N}\mathbf 1_N^{\mathsf T} \\
        \frac1{\sqrt N}\mathbf 1_N
        &
        \Id_N-\frac1N\mathbf 1_N\mathbf 1_N^{T}
    \end{pmatrix}.
\end{equation}
where \(\mathbf 1_N\in\mathbb C^N\) is the all-ones column vector and \(\Id_N\) is the identity matrix of size \(N\times N\).
\end{lemma}
\begin{proof}
    See \Cref{Proof:S-Matrix}
\end{proof}
\begin{corollary}
    The unitary operator $S$ is Hermitian and $S=S^\dagger=S^{-1}$.
\end{corollary}

\paragraph{The Recording Oracle} We define the global recording transform $\mathcal{T}$ and recording query operator $\mathcal{R}$ over $\mathcal{H}_{\Alg} \otimes \mathcal{H}_\rec$ such that
\[
\mathcal{T}= \mathbb{I}_{QPW} \otimes \left((S\otimes S)^{\otimes B}\right), 
\qquad 
\mathcal{R}=\mathcal{T}^\dagger \mathcal{O}\mathcal{T}=\mathcal{T} \mathcal{O}\mathcal{T}.
\]

\begin{theorem}[Theorem 3.3, \cite{HamoudiMagniez2023}, adapted]\label{thm:relation-phi-psi}
    Let $r=(r_t)_{t\in [B]}$ be the input sampled from the distribution $\bigotimes_{t\in [B]}\text{Uniform}([N]\times [N])$ and $(U_0, U_1, \ldots, U_T)$ be a $T$-query quantum algorithm. Set up $\ket{\psi_T}$ and $\ket{\phi_T}$ to be the joint state in the standard query model and the recording query model, respectively, where
    \[
    \ket{\psi_T}
    =
    \left(U_T\otimes \mathbb{I}_{\inp}\right) 
    \mathcal{O}
    \left(U_{T-1}\otimes \mathbb{I}_{\inp}\right) 
    \ldots 
    \mathcal{O}\left(U_0\otimes \mathbb{I}_{\inp}\right) 
    \ket{0}_{\Alg}  \ket{\mathcal{D}}_{\inp}, 
    \]
    \[
    \ket{\phi_T}
    =
    \left(U_T\otimes \mathbb{I}_{\rec}\right) 
    \mathcal{R} 
    \left(U_{T-1}\otimes \mathbb{I}_{\rec}\right) 
    \ldots 
    \mathcal{R}
    \left(U_0\otimes \mathbb{I}_{\rec}\right) 
    \ket{0}_{\Alg}  \left(\ket{\perp\perp}^{\otimes  B}\right)_{\inp}.
    \]
    Then they satisfy that
    \[
    \ket{\psi_T} = \mathcal{T} \ket{\phi_T}.
    \]
\end{theorem}
\begin{proof}
    For sake of completeness, the proof is included in \Cref{append:Proof-of-psi-phi}.
\end{proof}
\subsection{Analysis of the Recording Progress}
Recall the recording space $\mathcal{H}_{\rec}=(\mathcal{H}_{\rec, \loc}\otimes \mathcal{H}_{\rec, \loc})^{\otimes B}$ where $\mathcal{H}_{\rec, \loc}$ is spanned by $\{\ket{\perp} \} \cup [N]$. Given a computational basis  $\ket{r}=\bigotimes_{t\in [B]} \ket{r_t}$ of the recording space with $r_t \in \{\perp\perp \}\cup [N]\times [N] \cup \{\perp\}\times [N]\cup [N]\times \{\perp\}$, we say the block $t\in [B]$ is vacuum if its local state $\ket{r_t}=\ket{\perp\perp}$. Otherwise, we say it is nonvacuum. 
Set up its weight
\[
    \wt(r)=\abs{\{t\in [B]:\ket{r_t}\ne \ket{\perp\perp}}
\]
as the number of nonvacuum blocks in $\ket{r}$. We define the following projectors to measure the recording progress.
\begin{definition}
    For every nonnegative integer $k\geq 0$, we define the projectors by giving the basis states on which they project:
    \begin{itemize}
        \item $\Pi_{=k}$, $\Pi_{\leq k}$ and $\Pi_{\geq k}$: all basis states $\ket{t, \tau, x,u,w}\ket{r}$ such that its weight $\wt(r)=k$, $\wt(r)\leq k$, and $\wt(r)\geq k$ respectively.
        \item $\Pi_{=k, \perp}:$ all basis states $\ket{t, \tau, x,u,w}\ket{r}$ such that $r$ has weights $\wt(r)=k$ and $\ket{r_t}=\ket{\perp\perp}$.
    \end{itemize}
\end{definition}
\begin{definition}[Progress Measurement]
We define the measure of progress $\Delta_{t, k}$ for $t$ queries and $k$ nonvacuum blocks as
    \[
    \Delta_{t,k}=\norm{\Pi_{\ge k}\ket{\phi_t}},
    \]
where $\ket{\phi_t}$ is the joint state after $r$ queries in the recording query model. 
\end{definition}
In this section, we first prove the recurrence formula of the progress $\Delta_{t,k}$ stated in \Cref{lem:progress-recurrence'} and then the upper bound of the progress $\Delta_{t,k}$ stated in \Cref{lem:binomial-progress'}.

\begin{lemma}
\label{lem:progress-recurrence'}
For all \(t,k\ge0\), it has $\Delta_{t+1,k+1}
    \le
    \Delta_{t,k+1}
    +
    \frac2{\sqrt N}\Delta_{t,k}.$
\end{lemma}
\begin{lemma}
\label{lem:binomial-progress'}
For all integers \(0\le k\le t\), the progress $\Delta_{t,k}
    \le
    \binom tk
    \left(\frac2{\sqrt N}\right)^k$.
\end{lemma}
\subsubsection{Proof of \Cref{lem:progress-recurrence'}}
By definition, 
    \begin{align*}
    \Delta_{t+1,k+1}=\norm{\Pi_{\geq k+ 1} \ket{\phi_{t+1}}}=\norm{\Pi_{\geq k+ 1} \mathcal{T}\ket{\psi_{t+1}}}
    &=\norm{\Pi_{\geq k+ 1} \mathcal{T}(U_{t+1}\otimes \mathbb{I})\mathcal{O}\ket{\psi_{t}}}\\
    &=
    \norm{\Pi_{\ge k+1}\left(U_{t+1}\otimes \mathbb{I}\right) \mathcal{R}\ket{\phi_t}}\\
    &=
    \norm{\left(U_{t+1}\otimes \mathbb{I}\right)\Pi_{\ge k+1}\mathcal{R}\ket{\phi_t}}\\
    &=
    \norm{\Pi_{\ge k+1}\mathcal{R}\ket{\phi_t}}.
\end{align*}
where the equations hold since $\left(U_{t+1}\otimes \mathbb{I}\right)$ is a unitrary acting over the algorithm register $QPW$ only and $\Pi_{\geq k+1}$ is a projector acting over the oracle register $F$ only.

For any basis state $\ket{t,\tau, x,u,w}\ket{r}$ in $\ket{\phi_t}$, it may contribute to $\Delta_{t+1, k+1}$ only when it has $\wt(r)\geq k+1$ or it has $\wt(r)=k$, $u \neq 0$ and $\ket{r_t}= \ket{\perp\perp}$. Therefore, by the triangle inequality,
\begin{align*}
    \Delta_{t+1,k+1}
    & \leq
    \norm{\Pi_{\ge k+1}\mathcal{R}\Pi_{\ge k+1}\ket{\phi_t}}
    +
    \norm{\Pi_{\ge k+1}\mathcal{R}\Pi_{=k,\perp}\ket{\phi_t}}\\
    & \leq \norm{\Pi_{\ge k+1}\mathcal{R}}\cdot \norm{\Pi_{\ge k+1}\ket{\phi_t}} +     \norm{\Pi_{\ge k+1}\mathcal{R}\Pi_{=k,\perp}}\cdot \norm{\Pi_{=k,\perp}\ket{\phi_t}}
    \\
    & \leq 1 \cdot \Delta_{t,k+1} +
    \norm{\Pi_{\ge k+1}\mathcal{R}\Pi_{=k,\perp}}\cdot \Delta_{t,k}.
\end{align*}
since $\norm{\Pi_{=k,\perp}\ket{\phi_t}}\leq \norm{\Pi_{\geq k}\ket{\phi_t}} =\Delta_{t, k}$.

In the following, we prove the following bound 
\begin{equation}\label{eq: bound 2-sqrt-N}
    \norm{\Pi_{\ge k+1}\mathcal{R}\Pi_{=k,\perp}} \leq \frac{2}{\sqrt{N}}.
\end{equation}

First, we note that
\[
\norm{\Pi_{\ge k+1}\mathcal{R}\Pi_{=k,\perp}} = \max_{\norm{\ket{\psi}}=1}\norm{\Pi_{\geq k+1} \mathcal{R}\Pi_{=k, \perp} \ket{\psi}} 
\]
Consider any basis state
\(
|z\rangle=|t,\tau,x,u,w\rangle|r\rangle
\)
such that $\wt(r)=k$ and $|r_t\rangle=|\perp\perp\rangle$. Clearly, it belongs
to the support of $\Pi_{=k,\perp}$.

We first justify that, in the present setting, it suffices to maximize over such
computational basis states. Let $|\psi\rangle$ be an arbitrary normalized state
in the support of $\Pi_{=k,\perp}$, and expand it in the computational basis as
\(
|\psi\rangle=\sum_z \alpha_z |z\rangle,
\qquad
\sum_z|\alpha_z|^2=1,
\)
where every $|z\rangle$ in the sum satisfies $\wt(r)=k$ and
$|r_t\rangle=|\perp\perp\rangle$.

Recall that $\mathcal R=\mathcal T\mathcal O\mathcal T$.
The operator $\mathcal T$ acts only on the input register, while
$\mathcal O$ leaves the algorithm basis state
$|t,\tau,x,u,w\rangle$ unchanged and acts only on the input block indexed by
$t$. Hence $\mathcal R$ also leaves the algorithm basis state unchanged and
acts only on block $t$ of the input register.

Consequently, for any two distinct basis states $|z\rangle$ and
$|z'\rangle$ occurring in the above expansion, the states
\[
\Pi_{\ge k+1}\mathcal R|z\rangle
\qquad\text{and}\qquad
\Pi_{\ge k+1}\mathcal R|z'\rangle
\]
are orthogonal. Indeed, if their algorithm-register parts are different, this
is immediate. Otherwise, they have the same queried block $t$, and since
$|r_t\rangle=|r'_t\rangle=|\perp\perp\rangle$, the two input basis states must
differ on some block other than $t$; such blocks are left unchanged by
$\mathcal R$.

Therefore,
\begin{align*}
\left\|\Pi_{\ge k+1}\mathcal R|\psi\rangle\right\|^2
=
\sum_z |\alpha_z|^2
\left\|\Pi_{\ge k+1}\mathcal R|z\rangle\right\|^2
\le
\max_z
\left\|\Pi_{\ge k+1}\mathcal R|z\rangle\right\|^2.
\end{align*}
Taking the supremum over all normalized $|\psi\rangle$ in the support of
$\Pi_{=k,\perp}$, and observing that every such basis state $|z\rangle$ is
itself an allowed normalized input, we obtain
\[
\left\|\Pi_{\ge k+1}\mathcal R\Pi_{=k,\perp}\right\|
=
\max_z
\left\|\Pi_{\ge k+1}\mathcal R|z\rangle\right\|.
\]
Set up an oracle $\mathcal{O}^\tau_{x,u}$ for block $t$ such that given any $i,j \in [N]$, 
\begin{align*}
    \mathcal{O}^\tau_{x,u}\ket{i,j}=\omega^{u\cdot f_t(\tau, x)} \ket{i, j}
\end{align*}
where $\omega=e^{2\pi i/N+1}$, 
$f_t(A, x)= (j +1) \cdot \mathbb{1}[x=i]$ and $f_t(B, x)= (i +1) \cdot \mathbb{1}[x=j]$.
Set up the value 
\[
\beta_u^{\perp} =\braket{\perpp\perpp|\left(S\otimes S\right) \mathcal{O}^\tau_{x,u} \left(S\otimes S\right) |\perpp\perpp}.
\]
We claim the followings two properties of $\beta_u^{\perp}$.
\begin{claim}\label{clm: beta meaning}
    $\norm{\Pi_{\geq k+1} \mathcal{R}\ket{z}}^2 = 1-\abs{\beta_u^\perp}^2$.
\end{claim}
\begin{claim}\label{clm: beta}
If $u\neq 0$, then $\beta_u^\perp =1-\frac{1}{N}-\frac{1}{N^2}$. If $u= 0$, then $\beta_u^\perp =1$.
\end{claim}
By \Cref{clm: beta meaning} and \Cref{clm: beta}, for any basis state $\ket{z}=\ket{t,\tau, x,u,w}\ket{r}$ with $\wt(r)=k$ and $\ket{r_t}=\ket{\perp\perp}$, 
\begin{align*}
    \norm{\Pi_{\geq k+1} \mathcal{R}\ket{z}}^2 = 1-|\beta_u^\perp|^2 \leq 1-\left(1-\frac{2}{N}\right)^2 = \frac{4}{N} - \frac{4}{N^2} \leq \frac{4}{N}.
\end{align*}
It implies that $\norm{\Pi_{\ge k+1}\mathcal{R}\Pi_{=k,\perp}} = \max_z \norm{\Pi_{\geq k+1} \mathcal{R}\ket{z}}\leq \frac{2}{\sqrt{N}}$. Overall, 
\begin{align*}
    \Delta_{t+1,k+1}
    & \leq 1 \cdot \Delta_{t,k+1} +
    \norm{\Pi_{\ge k+1}\mathcal{R}\Pi_{=k,\perp}}\cdot \Delta_{t,k}
    \leq \Delta_{t,k+1} +
    \frac{2}{\sqrt{N}} \Delta_{t,k}.
\end{align*}
This concludes the proof of \Cref{lem:progress-recurrence'}.

\begin{proof}[Proof of \Cref{clm: beta meaning}]
    Given that $\mathcal{R}=\mathcal{T}\mathcal{O}\mathcal{T}$ and $T=\mathbb{I}_{PQW}\otimes  \left(\left(S\otimes S\right)^{\otimes B}\right)$, consider any basis state $\ket{z}=\ket{t, \tau, x,u,w}\ket{r}$ such that $\wt(r)=k$ and $\ket{r_t}=\ket{\perp\perp}$.
    Then,
    \[
    \mathcal{R}\ket{z} = \ket{t,\tau, x, u,w}\ot \left(\ket{z_t} \otimes \left(\bigotimes_{t'\in [B]: t' \neq t} \ket{r_{t'}}\right)\right)
    \]
    where $ \ket{z_t}=\left(S\otimes S\right) \mathcal{O}^\tau_{x,u} \left(S\otimes S\right) \ket{\perp\perp}$. Note that for any basis state in $\mathcal{R}\ket{z}$, it either has weight exactly $k+1$ or it has weight exactly $k$. Therefore, we claim that 
    \[
    \norm{\Pi_{\geq k+1} \mathcal{R}\ket{z}}^2 + \norm{\Pi_{= k} \mathcal{R}\ket{z}}^2 = 1.
    \]
    In addition, for any basis state in $\mathcal{R}\ket{z}$ with weight exactly $k$, the block $t$ must remain the vacuum state. Therefore,
    \[
    \norm{\Pi_{= k} \mathcal{R}\ket{z}} = \braket{\perpp\perpp|z_t} = \beta_u^\perp, \qquad \norm{\Pi_{\geq k+1} \mathcal{R}\ket{z}}^2 = 1-|\beta_u^\perp|^2.
    \]
\end{proof}

\begin{proof}[Proof of \Cref{clm: beta}]
Recall that $\ket{\tilde{u}}=S\ket{\perp}=\frac{1}{\sqrt{N}}\sum_{i\in [N]}\ket{i}$. W.l.o.g., assume $\tau=A$. 
\begin{align*}
    \beta_u^{\perp}  & = \braket{\perpp\perpp|\left(S\otimes S\right) \mathcal{O}^\tau_{x,u} \left(S\otimes S\right)|\perpp\perpp} = \braket{\tilde{u}, \tilde{u}| \mathcal{O}^\tau_{x,u} |\tilde{u}, \tilde{u}} \\
    & = \frac{1}{N^2} \sum_{k,l\in [N]}\sum_{i,j\in [N]}\braket{k,l|\mathcal{O}^\tau_{x,u}|i,j}
    = 
    \frac{1}{N^2} \sum_{k,l\in [N]} \sum_{i,j\in [N]}\braket{k,l|\omega^{u\cdot f_t(\tau, x)}|i,j}\\
    & 
    = 
    \frac{1}{N^2} \sum_{i,j\in [N]} \omega^{u\cdot (j+1)\mathbb{1}[x=i]}
    = \frac{1}{N^2} \left(\sum_{j \in [N]} \omega^{u \cdot (j+1)} + \sum_{i \in [N]: i\neq x}\sum_{j \in [N]} \omega^{0} \right)
    \\
    & =\frac{1}{N^2} 
    \left(
    \sum_{k=1}^{N}\omega^{uk}+ N(N-1)
    \right).
\end{align*}
Note that $\omega=e^{\frac{2\pi i}{N+1}}$ is a primitive $(N+1)$-th root of unity. If $u \in \mathbb{Z}_{N+1}$ and $u \neq 0$, then 
\[
\sum_{k=0}^{N} \omega^{uk} = 0 \quad \text{ and } \quad \sum_{k=1}^{N} \omega^{uk} = 0 - \omega^{0}=-1. 
\]
Therefore, if $u=0$, then $\beta^\perp_u=1$. If $u \neq 0$, then
\[
\beta_u^\perp = \frac{1}{N^2} \left(-1 + N^2 -N\right)= 1-\frac{1}{N}-\frac{1}{N^2}.
\]
\end{proof}

\subsubsection{Proof of \Cref{lem:binomial-progress'}}
\begin{proof}[Proof of \Cref{lem:binomial-progress'}]
    Consider the base case where $\Delta_{0,0}=1$ and $\Delta_{0,k}=0$ for all $k\geq 1$. In addition, $\Delta_{t, k}=0$ for every $k\geq t+1$. 
    Consider the general case with $t\geq k\geq 1$. By \Cref{lem:progress-recurrence'},
    \[
    \Delta_{t,k}
    \leq 
    \Delta_{t-1,k} + \frac{2}{\sqrt{N}} \Delta_{t-1,k-1}.
    \]
    If $t=k$, by induction hypothesis, $\Delta_{t,k}
    \leq 
     \frac{2}{\sqrt{N}} \Delta_{t-1,k-1} \leq \binom{t-1}{k-1}\left(\frac{2}{\sqrt{N}}\right)^{k}
     \leq
     \binom{t}{k}\left(\frac{2}{\sqrt{N}}\right)^{k}$.
    If $t\geq k+1$, by \Cref{lem:progress-recurrence'} and induction hypothesis, 
    \begin{align*}
        \Delta_{t,k}
    \leq 
    \Delta_{t-1,k} + \frac{2}{\sqrt{N}} \Delta_{t-1,k-1} 
    & \leq 
    \binom{t-1}{k} \left(\frac{2}{\sqrt{N}}\right)^k + \left(\frac{2}{\sqrt{N}}\right) \cdot\binom{t-1}{k-1} \left(\frac{2}{\sqrt{N}}\right)^{k-1} = \binom{t}{k} \left(\frac{2}{\sqrt{N}}\right)^k.
    \end{align*}
where the last equation holds by Pascal's identity.
\end{proof}
\subsection{Soundness of Recording}
In this section, we prove the main soundness \Cref{lem:technical lemma 1} that gives low success probability for low-weight records.

Recall that, for each block $t\in [B]$, its local space is denoted by $\mathcal{H}_t=\mathcal{H}_{\rec,\loc} \otimes \mathcal{H}_{\rec,\loc}$ where $\mathcal{H}_{\rec,\loc}$ is spanned by $\{\ket{\perp}, \ket{0}, \ldots, \ket{N-1}\}$.
Set up the vacuum projector $P_{t,\perp}$ and the non-vacuum projector $P_{t,\neq\perp}$ such that
\begin{equation}\label{eq:Projection perp}
        P_{t,\perp}=\proj{\perpp\perpp},
    \qquad
    P_{t,\neq\perp}=\Id-P_{t,\perp}.
\end{equation}
Given any integer $k \in \{0,1,\ldots, B\}$, set up 
\begin{equation}\label{eq:Lambda-k}
    \Lambda_{\le k}
    =  \sum_{\substack{L\subseteq[B]:\\|L|\le k}} \Lambda_L
    \qquad
    \text{where } \Lambda_L= \left(\bigotimes_{t\in L}P_{t,\neq \perp}\right)
    \otimes
    \left(\bigotimes_{t\notin L}P_{t,\perp}\right).
\end{equation}
Note that  $\Pi_{\leq k} = \mathbb{I}_{PQW} \otimes \Lambda_{\leq k}$.
Given a collection of nonzero orthogonal projectors $(Q_t)_{t\in D}$ with some nonempty $D \subseteq [B]$, for each $t\in D$, set
\[
p_t = \norm{Q_t \ket{\perpp\perpp}}^2 = \braket{\perpp\perpp|Q_t|\perpp\perpp}.
\]
Set up the projector 
\[
    Q_D=
    \left(\bigotimes_{t\in D}Q_t\right)
    \otimes
    \left(\bigotimes_{t\in [B]\setminus D}\Id_{\mathcal H_t}\right).
\]
In this section, prove the following lemma.
\begin{lemma}\label{lem:technical lemma 1}
For every $t\in [B]$, let $\mathcal{H}_t= \mathcal{H}_{\rec,\loc}\otimes \mathcal{H}_{\rec,\loc}$ be the local space and $\ket{\perpp\perpp}$ be the vacuum state.
Let $\{Q_t\}_{t\in D\subseteq [B]}$ be a given collection of nonzero orthogonal projector with $p_t = \norm{Q_t \ket{\perpp\perpp}}^2$. Set up \(\{X_t\}_{t\in D}\) to be mutually independent Bernoulli random variables with $\Prb[X_t=0]=p_t$ and $\Prb[X_t=1]=1-p_t$. Given any integer $k\in \{0,1,\ldots, B\}$, it has
    \[
    \norm{Q_D\Lambda_{\le k}}^2
    = 
    \Prb\left[\sum_{t\in D}X_t\le k\right].
\]
\end{lemma}
\subsubsection{Proof of \Cref{lem:technical lemma 1}}
Consider the operator $Q_D\Lambda_{\leq k}$. Since $Q_D$ and $\Lambda_{\leq k}$ are orthogonal projectors over $\mathcal{H}_{\rec}$, then
\[
\norm{Q_D\Lambda_{\leq k}}^2 = \norm{(Q_D\Lambda_{\leq k})(Q_D\Lambda_{\leq k})^\dagger} = 
\norm{
Q_D \Lambda_{\leq k} Q_D}.
\]
Since $Q_D \Lambda_{\leq k} Q_D:\Ran(Q_D) \rightarrow \Ran(Q_D)$ is positive semidefinite, then
\[
\norm{Q_D \Lambda_{\leq k} Q_D} = \lambda_{\max}(Q_D \Lambda_{\leq k} Q_D)
\]
where $\lambda_{\max}(Q_D \Lambda_{\leq k}  Q_D) = \max \{\lambda: Q_D \Lambda_{\leq k}  Q_D  \ket{\phi}=\lambda\ket{\phi}, \text{for some }\ket{\phi}\in \Ran(Q_D) \text{ with }\ket{\phi}\neq 0\}$. 

\paragraph{Orthogonal Decomposition}
To compute $\lambda_{\max}(Q_D \Lambda_{\leq k}  Q_D)$, we first consider an orthogonal decomposition of $\Ran(Q_D)$. For each block $t\in [B]$, let $\mathcal{H}_t^{\neq \perp}= \{\ket{\phi}\in \mathcal{H}_t: \braket{\perpp\perpp|\phi}=0\}$ be the subspace orthogonal to the vacuum state $\ket{\perpp\perpp}$.
For each block $t\in D$ with $p_t >0$, set up the unit vector $\ket{w_t}=\frac{1}{\sqrt{p_t}}Q_t\ket{\perpp\perpp}$. We claim the following lemma.
\begin{lemma}\label{lem:Ran(Q) orthogonal decomp}
For each block $t\in D$ with a given projector $Q_t$ such that $p_t >0$, it has the orthogonal decomposition that \[
\Ran(Q_t)=\spanop\{\ket{w_t}\} \oplus R_t, \quad \text{where } R_t=\Ran(Q_t)\cap \mathcal{H}^{\neq \perp}_t.
\]
\end{lemma}
Before proving \Cref{lem:Ran(Q) orthogonal decomp}, we first see its implications. 
Recall $Q_D = \left(\bigotimes_{t\in D} Q_t \right)\otimes \mathbb{I}_{[B]\setminus D}$. Let $Z=\{t\in D: p_t=0\}$. For each block $t\in Z$, 
\[
\Ran(Q_t) = R_t = \Ran(Q_t) \cap \mathcal{H}_t^{\neq \perp}.
\]
For each block $t\in D\setminus Z$, by \Cref{lem:Ran(Q) orthogonal decomp}, it has $\Ran(Q_t)=\spanop\{\ket{w_t}\} \oplus R_t$. For each remaining block $t\in [B]\setminus D$, the support trivially decomposes as $\spanop\{\ket{\perpp\perpp}\} \oplus \mathcal{H}_t^{\neq \perp}$. Therefore,
\begin{equation}\label{eq:Ran(Q_D)}
    \Ran(Q_D)=
\left(\bigotimes_{t\in Z} R_t\right) 
\otimes 
\left(\bigotimes_{t\in D\setminus Z} \left(\spanop\{\ket{w_t}\} \oplus R_t\right)\right) 
\otimes
\left(\bigotimes_{t\in [B]\setminus D}\left(\spanop\{\ket{\perpp\perpp}\} \oplus \mathcal{H}_t^{\neq \perp}\right) \right) 
\end{equation}
By rearranging \Cref{eq:Ran(Q_D)}, we get the following conclusion.
\begin{lemma}\label{lem:orthogonal decomp Ran(Q_D)}
    The range of the projector $Q_D$ has an orthogonal decomposition,
    \[
     \operatorname{Ran}(Q_D)
    =\bigoplus_{\substack{Z\subseteq J\subseteq D\\S\subseteq[B]\setminus D}}
    \mathcal H_{J,S}
    \]
where $\mathcal H_{J,S}
    =
    \left(\bigotimes_{t\in J}R_t\right)
    \otimes
    \left(\bigotimes_{t\in D\setminus J}
          \operatorname{span}\{\ket{w_t}\}\right)
    \otimes
    \left(\bigotimes_{t\in S}\mathcal H_t^{\neq \perp}\right)
    \otimes \left(\bigotimes_{t\in[B]\setminus(D\cup S)}
          \operatorname{span}\{\ket{\perpp\perpp}\}\right)$.
\end{lemma}
\begin{proof}[Proof of \Cref{lem:Ran(Q) orthogonal decomp}]
First, we prove the two subspaces $\spanop\{w_t\}$ and $R_t$ are orthogonal. Recall that $\ket{w_t}=\frac{1}{\sqrt{p_t}}Q_t\ket{\perpp\perpp}$, then $Q_t \ket{w_t} = \ket{w_t}$. For any $\ket{\phi}\in R_t$, \[
\braket{w_t|\phi}=\frac{1}{\sqrt{p_t}}\braket{\perpp\perpp|Q_t|\phi}=\frac{1}{\sqrt{p_t}}\braket{\perpp\perpp|\phi}=0.
\]
Since $\ket{w_t}\in \Ran(Q_t)$ and $R_t\subseteq \Ran(Q_t)$, we have that $\spanop\{\ket{w_t}\} \oplus R_t \subseteq \Ran(Q_t)$.
\\~\\
On the other hand, for any state $\ket{\phi}\in \Ran(Q_t)$, set $\alpha= \braket{w_t|\phi}$ and $\ket{\phi'}=\ket{\phi}-\alpha \ket{w_t}$. Note that $\ket{\phi'}\in R_t$ since $Q_t \ket{\phi'}=\ket{\phi'}$ and $\braket{\perpp\perpp|\phi'}=\braket{\perpp\perpp|Q_t|\phi}-\alpha\braket{\perpp\perpp|Q_t|w_t}=\sqrt{p_t}\braket{w_t|\phi}-\alpha\sqrt{p_t}\braket{w_t|w_t}=0$. Therefore, $\spanop\{\ket{w_t}\} \oplus R_t = \Ran(Q_t)$.
\end{proof}

\paragraph{Eigenvalue Computation} In this paragraph, we go through each subspace $\mathcal{H}_{J,S}$ and compute the eigenvalue of $Q_D \Lambda_{\leq k} Q_D$. We prove the following main result.

\begin{lemma}\label{lem:technical-lambda-JS}
For every $Z\subseteq J\subseteq D$ and $S\subseteq[B]\setminus D$, the operator $Q_D\Lambda_{\le k}Q_D$ acts on $\mathcal H_{J,S}$ as multiplication by
\[
\lambda_{J,S}
=
\sum_{\substack{L'\subseteq D\setminus J\\|L'|+|J|+|S|\le k}}
\left(\prod_{t\in L'}(1-p_t)\right)
\left(\prod_{t\in D\setminus(J\cup L')}p_t\right).
\]
\end{lemma}

\begin{proof}
We start by noting that for $t\in D$ with $p_t>0$,
\[
Q_tP_{t,\perp}Q_t=p_t\proj{w_t},
\qquad
Q_tP_{t,\neq\perp}Q_t=Q_t-p_t\proj{w_t}.
\]
and with $p_t=0$,
\[
Q_tP_{t,\perp}Q_t=0, \qquad Q_tP_{t,\neq\perp}Q_t=Q_t
\]
Thus, for $\ket{r_t}\in R_t$,
\begin{align}
Q_tP_{t,\perp}Q_t\ket{r_t}&=0,
&
Q_tP_{t,\neq\perp}Q_t\ket{r_t}&=\ket{r_t},
\label{eq:r_t}
\end{align}
and
\begin{align}
Q_tP_{t,\perp}Q_t\ket{w_t}&=p_t\ket{w_t},
&
Q_tP_{t,\neq\perp}Q_t\ket{w_t}&=(1-p_t)\ket{w_t}.
\label{eq: w_t}
\end{align}
Now, for $L\subseteq[B]$, set
$T_L=Q_D\Lambda_LQ_D=\bigotimes_{t\in[B]}T_{L,t}$, where
\begin{equation}\label{eq:T-L-t cases}
T_{L,t}
=
\begin{cases}
Q_tP_{t,\neq\perp}Q_t,&t\in D\cap L,\\
Q_tP_{t,\perp}Q_t,&t\in D\setminus L,\\
P_{t,\neq\perp},&t\in L\setminus D,\\
P_{t,\perp},&t\in[B]\setminus(D\cup L).
\end{cases}
\end{equation}
It suffices to check the action on elementary tensors spanning $\mathcal H_{J,S}$. For such a tensor, $T_L$ acts as 0 unless 
\[
J\subseteq L
\qquad\text{and}\qquad
L\cap([B]\setminus D)=S.
\]
Writing $L=J\sqcup S\sqcup L'$ with $L'\subseteq D\setminus J$, the resulting scalar is
\[
\left(\prod_{t\in J}1\right)
\left(\prod_{t\in L'}(1-p_t)\right)
\left(\prod_{t\in D\setminus(J\cup L')}p_t\right).
\]
Summing over all $L$ with $|L|\le k$ proves the claim, and linearity extends the conclusion from elementary tensors to every vector in $\mathcal H_{J,S}$.
\end{proof}
\paragraph{Remaining Proof of \Cref{lem:technical lemma 1}} By definition, 
\[
\norm{Q_D\Lambda_{\leq k}}^2  = 
\lambda_{\max}(Q_D \Lambda_{\leq k}  Q_D)
\]
By \Cref{lem:orthogonal decomp Ran(Q_D)}, we have the orthogonal decomposition of $\Ran(Q_D)$ such that 
\[
    \operatorname{Ran}(Q_D)
    =\bigoplus_{\substack{Z\subseteq J\subseteq D\\S\subseteq[B]\setminus D}}
    \mathcal H_{J,S}
\]
and by \Cref{lem:technical-lambda-JS}, every subspace $\mathcal H_{J,S}$ has the eigenvalue $\lambda_{J,S}$. Set up independent Bernoulli variables $\{X_t\}_{t\in X}$ where $\Prb[X_t=1]=1-p_t$. Then,
\[
\lambda_{J,S}
=
\sum_{\substack{L'\subseteq D\setminus J: \\ |L'|+|J|+|S|\leq k}} \left(\prod_{t\in L'}(1-p_t)\right)\left(\prod_{t\in D\setminus (L'\cup J)}p_t\right) 
=
\Prb\left[\sum_{t\in D\setminus J}X_t \leq k-|J|-|S|\right].
\]
Given that $\sum_{t\in D} X_t 
\leq 
|J|+\sum_{t\in D\setminus J} X_t 
\leq
|S|+|J|+\sum_{t\in D\setminus J} X_t 
$, then every eigenvalue 
\[
\lambda_{J,S}=\Prb\left[\sum_{t\in D\setminus J}X_t+|J|+|S| \leq k\right]
\leq 
\Prb\left[\sum_{t\in D}X_t \leq k\right].
\]
In addition, set $J=Z= \{t\in D:p_t=0\}$ and $S=\emptyset$. Consider the subspace $\mathcal{H}_{Z,\emptyset}$. Note that every $t\in Z$ has $p_t=0$ and $X_t=1$. Then, we have
\[
\lambda_{Z,\emptyset}= \Prb\left[|Z|+\sum_{t\in D\setminus Z} X_t\leq k\right] = \Prb\left[\sum_{t\in D} X_t\leq k\right].
\]
Therefore, $\norm{Q_D\Lambda_{\leq k}}^2 
= 
\lambda_{\max}(Q_D \Lambda_{\leq k} Q_D)
= \lambda_{Z, \emptyset} = \Prb\left[\sum_{t\in D} X_t\leq k\right]$.


\subsection{From the Recording Progress to the Success Probability}

Let $r=(r_t^A, r_t^B)_{t\in [B]}$ be an input sampled from $\bigotimes_{t\in [B]} \text{Uniform}([N]\times [N])$. The task $(B,N,d)$-$\bhvf$ is to output the hidden values $(r_t^A,r_t^B)$ for at least $\lceil dB \rceil$ distinct blocks. 
For every computational basis $\ket{\eta}= \ket{t,\tau, x,u,w}_{QPW}$ in the algorithm space $\mathcal{H}_\Alg$, the working register $W$ holds some value $w$ that presents a classical output of the algorithm. Specifically, each basis $\ket{\eta}\in \mathcal{H}_\Alg$ is said to be well formed if it names
\begin{itemize}
    \item a block subset $D_\eta \subseteq [B]$ with $|D_{\eta}|\geq \lceil dB \rceil$ and
    \item a guessing value $(t,g^A_t, g^B_t)_{t\in D_\eta}$ for each block $t \in D_\eta$.
\end{itemize}
Given an input $r=(r_t^A, r^B_t)_{t\in [B]}$, 
we say a computational basis output $\ket{\eta}\in \mathcal{H}_\Alg$ is successful on it if $g^A_t=r^A_t$ and $g^B_t=r^B_t$ for every $t\in D_\eta$.
For each well-formed basis $\ket{\eta}$, we define the projector $Q_\eta$ such that
\[
    Q_\eta
    =
    \left(
    \bigotimes_{t\in D_\eta}
    \proj{g^A_t, g_t^B}
    \right)
    \otimes 
    \left(
    \bigotimes_{t\in [B]\setminus D_\eta}
    \Id_{\mathcal{H}_t}
    \right).
\]
For each ill-formed basis $\ket{\eta}$ (violating the output syntax or naming less than $\lceil dB \rceil$ blocks), we define $Q_\eta = 0$. The overall success projector is
\begin{equation}\label{eq: pi-success}
        \Pi_{succ}^{BHVF}
    =
    \sum_{\eta}
    \proj{\eta}_{\Alg}\otimes Q_\eta.
\end{equation}
We claim the following lemma and use it to prove \Cref{thm:bhvf}.
\begin{lemma}\label{lem:main-mbeh helper'}
For every basis $\ket{\eta}\in \mathcal{H}_\Alg$, given any integer $k \leq dB/2$, 
\[
    \norm{Q_\eta\mathcal{T}_B\Lambda_{\le k}}^2
    \le
    e^{-dB/24}
\]
where $\mathcal{T}_B=(S\otimes S)^{\otimes B}$ and $\Lambda_{\le k}
    =  \sum_{\substack{L\subseteq[B]:\\|L|\le k}} \Lambda_L$ is defined by \Cref{eq:Projection perp} and \Cref{eq:Lambda-k}.
\end{lemma}
\subsubsection{Proof of \Cref{lem:main-mbeh helper'}}
If \(\ket{\eta}\) is ill-formed, then \(Q_\eta=0\) and the claim is correct trivially. 
Assume \(\ket{\eta}\) is well-formed. Let $m=|D_\eta|$.
For each selected block \(t\in D_\eta\), define $\ket{z_t}=(S\otimes S)\ket{g_t^A, g_t^B}$ where $g^A_t, g^B_t \in [N]$. Then, 
\[
    \mathcal{T}_B Q_\eta\mathcal{T}_B
    = 
    (S\otimes S)^{\otimes B}Q_\eta  (S\otimes S)^{\otimes B}
    =
    \left(\bigotimes_{t\in D_\eta}\proj{z_t}
    \right)\otimes \Id_{[B]\setminus D_\eta}.
\]
For each $t\in D_\eta$, set up the projector $Q_t=\proj{z_t}$ and $p_t = \norm{Q_t \ket{\perpp\perpp}}^2$ where
\[
p_t 
= \bra{\perpp\perpp}Q_t\ket{\perpp\perpp}
=\abs{\braket{\perpp\perpp|z_t}}^2 =\abs{\braket{\perpp\perpp|S \otimes S|g_t}}^2 =\abs{\braket{\tilde{u}, \tilde{u}|g^A_t, g^B_t}}^2=\abs{\frac{1}{\sqrt{N}} \cdot \frac{1}{\sqrt{N}}}^2=\frac{1}{N^2}
\]
since $\ket{\tilde{u}}=\frac{1}{\sqrt{N}}\sum_{i\in [N]}\ket{i}$. Set up \(\{X_t\}_{t\in D'_\eta}\) to be mutually independent Bernoulli random variables with $\Prb[X_t=0]=p_t$ and $\Prb[X_t=1]=1-p_t$. 
Given any integer $k\in \{0,1,\ldots, B\}$, by \Cref{lem:technical lemma 1}, 
\[
    \norm{Q_\eta\mathcal{T}_B \Lambda_{\le k}}^2
    =
    \norm{\mathcal{T}_B Q_\eta\mathcal{T}_B \Lambda_{\le k}}^2
    =
    \Prb\left[\sum_{t\in D_\eta}X_t\le k\right].
\]
Since $N\geq 2$, every random variable $X_t$ with $t\in D_\eta$ has $\Prb[X_t=1]=1-\frac{1}{N^2}\geq 3/4$. Set $\mu =\mathbb{E}[\sum_{t\in D_\eta} X_t]\geq 3m/4$ where $m=|D_\eta|\geq dB$.

For any $k\leq dB/2$, it has $k \leq \frac{1}{2}dB \leq \frac{2}{3}\mu$. Note that $\mu \geq 3dB/4$.
By Chernoff lower-tail bounds, 
\[
\Pr\left[
\sum_{t\in D_\eta} X_t \le k
\right]
\le
\Pr\left[
\sum_{t\in D_\eta} X_t
\le
\left(1-\frac{1}{3}\right)\mu
\right]
\le
\exp\left(
-\frac12\left(\frac13\right)^2\mu
\right)
\le e^{-dB/24}.
\]
Therefore $\norm{Q_\eta\mathcal T\Pi_{\le k}}^2
    \le
    e^{-dB/24}$.

\subsubsection{Proof of \Cref{thm:bhvf}}
It suffices to prove that, for every fixed constant \(d\in(0,1]\), there exist some constants \(c_d,b_d>0\) such that, for every \(B\ge b_d\) and \(N\ge2\), any quantum algorithm that solves $(B,N,d)$-$\bhvf$ with success probability at least \(1/3\) requires at least $c_d B\sqrt N$ quantum queries. Note that the success probability is taken over the uniform input distribution and the algorithm. Given an arbitrary constant $d \in (0,1]$, set up 
\begin{equation}\label{eq:parameter}
        k=\lceil dB/2\rceil,
    \qquad
    c_d=d/8e
\end{equation}
Consider any algorithm with $T< c_d B \sqrt{N}$ quantum queries. In the following, we prove that it succeeds with probability $< 1/3$. By the triangle inequality, 
\begin{align*}
\Prb[\text{success}]^{1/2} 
& = \norm{\Pi_{succ}^{BHVF}\mathcal T\ket{\phi_T}} 
\leq
\norm{\Pi_{succ}^{BHVF}\mathcal T\Pi_{\leq k-1}\ket{\phi_T}}+ \norm{\Pi_{succ}^{BHVF}\mathcal T\Pi_{\geq k}\ket{\phi_T}} 
\end{align*}
By \Cref{lem:binomial-progress'} and \Cref{eq:parameter}, for $T\geq k$
\begin{equation}\label{eq: ub 2-k}
    \norm{\Pi_{\ge k}\ket{\phi_T}}=
\Delta_{T, k}\leq \binom{T}{k} \left(\frac{2}{\sqrt{N}}\right)^k 
\leq 
\left(\frac{eT}{k}\cdot\frac{2}{\sqrt{N}}\right)^k 
\leq 
\left(\frac{e \cdot c_d B \sqrt{N} }{dB/2}\cdot\frac{2}{\sqrt{N}}\right)^k 
= \left(\frac{ 4e  c_d}{d}\right)^k 
= 2^{-k}.
\end{equation}
Recall the definition, $\mathcal{T}=\mathbb{I}_{QPW}\otimes \mathcal{T}_B$, $\Pi_{\leq k} = \mathbb{I}_{QPW} \otimes \Lambda_{\leq k}$ and $\Pi_{succ}^{BHVF} = \sum_{\eta} \proj{\eta}_{\Alg}\otimes Q_\eta$. 
With $k-1\leq dB/2 $,
\begin{equation}\label{eq:ub db-24}
    \norm{\Pi_{succ}^{BHVF}\mathcal T\Pi_{\leq k-1}\ket{\phi_T}}
    \leq    \norm{\Pi_{succ}^{BHVF}\mathcal T\Pi_{\leq k-1}}
    =
    \norm{Q_\eta \mathcal{T}_B \Lambda_{\leq k-1}}
    \leq 
    e^{-dB/48}.
\end{equation}
Combining \Cref{eq: ub 2-k} and \Cref{eq:ub db-24}, for any $B\geq b_d =25 \ln  3/d$, \begin{align*}
\Prb[\text{success}]^{1/2} 
& \leq
\norm{\Pi_{succ}^{BHVF}\mathcal T\Pi_{\leq k-1}\ket{\phi_T}}+ \norm{\Pi_{succ}^{BHVF}\mathcal T\Pi_{\geq k}\ket{\phi_T}} \\
& \leq \max_{\eta} \norm{Q_\eta \mathcal{T}_B\Lambda_{\leq k-1}} + \norm{\Pi_{succ}^{BHVF}\mathcal T} \cdot  \norm{\Pi_{\geq k}\ket{\phi_T}} \\
& \leq e^{-dB/48} + 2^{-dB/2} \\
& <\frac{1}{\sqrt{3}}.
\end{align*}
Therefore, any quantum algorithm that succeeds with probability at least $1/3$ requires at least $c_d B\sqrt N$ queries. 


\bibliographystyle{alpha}
\bibliography{ref}

\appendix

\section{Lower Bounds of the RMIF Problem}\label{sec:appnd:RMIF-lb}
\paragraph{The Standard Query Model}
We use $\mathcal H_{\inp}$ to denote the input space over all $B$ blocks such that
\[
\mathcal H_{\inp}=\left(\mathcal H_{\inp,\loc}\right)^{\otimes B} \qquad \text{where }     \mathcal H_{\inp,\loc}=\spanop\{\ket{0}, \ket{1},\ldots, \ket{{N-1}}\}.
\]
The algorithm space $\mathcal{H}_\Alg$ is over three registers $Q,P,W$, where the query register $Q$ holds a block index $t\in [B]$ and a position index $x\in [N]$, the phase register $P$ holds $u\in \mathbb{Z}_{2}$
and the working register $W$ holds some value $w$. The algorithm space \(\mathcal H_{\Alg}\) is spanned by basis states $\ket{t,x}_Q\ket{u}_P\ket{w}_W$.
Given an input $r=(r_t)_{t\in [B]}$ with $r_t\in [N]$, the query oracle $O_r$ is defined by
\[
\mathcal{O}_r \ket{t,x,u,w} = 
    \omega^{u \cdot f_t(x)} \ket{t,x,u,w} 
\quad \text{where } \omega=e^{\pi i}=-1.
\]
Recall that $f_t(x)=\mathbbm{1}[x=r_t]$. 
\paragraph{Quantum Query Algorithms in the Standard Model}
Consider any $T$-query algorithm specified by a sequence of $U_0, U_1, \ldots, U_T$ of unitary operators. 
Note that the input $r=(r_t)_{t\in [B]}$ is sampled from $\mathcal{D}_{B,N}=\bigotimes_{t\in [B]}\text{Uniform}([N])$.
The joint state $\ket{\psi_t}$ over the joint space $\mathcal{H}_{\Alg}\otimes \mathcal{H}_\inp$ after $t\leq T$ queries to some input $r$ is 
\begin{equation}\label{eq:joint-state-standard}
    \ket{\psi_t}
    =
    \left(U_t\otimes \mathbb{I}_{\inp}\right) 
    \mathcal{O} 
    \left(U_{t-1}\otimes \mathbb{I}_{\inp}\right) 
    \ldots 
    \mathcal{O}\left(U_0\otimes \mathbb{I}_{\inp}\right) 
    \ket{0}_{\Alg}  \ket{\mathcal{D}}_{\inp}
\end{equation}
where the initial input state is $\ket{\mathcal{D}}_\inp= \ket{\tilde{u}}^{\otimes B}$ with $ \ket{\tilde{u}}=\frac1{\sqrt N}\sum_{i\in [N]}\ket{i}$ and $\mathcal{O}=\sum_r \mathcal{O}_r \otimes \proj{r}$ is a unitary over the joint space $\mathcal{H}_{\Alg}\otimes \mathcal{H}_\inp$ .

\paragraph{The Recording Query Model} Let $\ket{\perp}$ be the additional state that is orthogonal to 
\\
$\spanop\{\ket{0}, \ket{1},\ldots, \ket{{N-1}}\}$. We use $\mathcal H_{\rec}$ to denote the recording space over all $B$ blocks such that
\[
\mathcal H_{\rec}=\left(\mathcal H_{\rec,\loc}\right)^{\otimes B} \qquad \text{where }     \mathcal H_{\rec,\loc}=\spanop\{\ket{\perp}, \ket{0}, \ket{1},\ldots, \ket{{N-1}}\}.
\]
We define the global recording transform $\mathcal{T}$ and recording query operator $\mathcal{R}$ over $\mathcal{H}_{\Alg} \otimes \mathcal{H}_\rec$ such that
\[
\mathcal{T}= \mathbb{I}_{QPW} \otimes \left(S^{\otimes B}\right), 
\qquad 
\mathcal{R}=\mathcal{T}^\dagger \mathcal{O}\mathcal{T}.
\]
where $S$ is a unitary operator defined by \Cref{def:local-swap-S}. Given the input $r=(r_t)_{t\in [B]}$ sampled from the distribution $\bigotimes_{t\in [B]}\text{Uniform}([N])$, by an analogous argument of \Cref{thm:relation-phi-psi}
    \[
    \ket{\psi_T} = \mathcal{T} \ket{\phi_T} 
    \]
where the joint state $\ket{\psi_T}$ in the standard query model is defined by \Cref{eq:joint-state-standard} and the joint state $\ket{\phi_T}$ in the recording query model is defined as
\[
    \ket{\phi_T}
    =
    \left(U_T\otimes \mathbb{I}_{\rec}\right) 
    \mathcal{R} 
    \left(U_{T-1}\otimes \mathbb{I}_{\rec}\right) 
    \ldots 
    \mathcal{R} 
    \left(U_0\otimes \mathbb{I}_{\rec}\right) 
    \ket{0}_{\Alg}  \left(\ket{\perp}^{\otimes  B}\right)_{\inp}.
    \]
\subsection{Analysis of the Recording Progress}
Given a computational basis  $\ket{r}=\bigotimes_{t\in [B]} \ket{r_t}$ of the recording space with $r_t \in \{\perp \}\cup [N]$, we say a block $t\in [B]$ is vacuum if $\ket{r_t}=\ket{\perp}$ and nonvacuum otherwise.
We define the weight
\[
    \wt(r)=\abs{\{t\in [B]:\ket{r_t}\ne \ket{\perp}}
\]
as the number of nonvacuum blocks in $\ket{r}$. We define the following projectors to measure the recording progress.
\begin{definition}
    For every nonnegative integer $k\geq 0$, we define the projectors by giving the basis states on which they project:
    \begin{itemize}
        \item $\Pi_{=k}$, $\Pi_{\leq k}$ and $\Pi_{\geq k}$: all basis states $\ket{t, x,u,w}\ket{r}$ such that its weight $\wt(r)=k$, $\wt(r)\leq k$, and $\wt(r)\geq k$ respectively.
        \item $\Pi_{=k, \perp}:$ all basis states $\ket{t, x,u,w}\ket{r}$ such that $r$ has weights $\wt(r)=k$ and $\ket{r_t}=\ket{\perp}$.
    \end{itemize}
\end{definition}
\begin{definition}[Progress Measurement]
We define the measure of progress $\Delta_{t, k}$ for $t$ queries and $k$ nonvacuum blocks as
    \[
    \Delta_{t,k}=\norm{\Pi_{\ge k}\ket{\phi_t}},
    \]
where $\ket{\phi_t}$ is the joint state after $r$ queries in the recording query model. 
\end{definition}
In this section, we first prove the analogous recurrence formula of the progress $\Delta_{t,k}$ stated in \Cref{lem:progress-recurrence'} and then the analogous upper bound of the progress $\Delta_{t,k}$ stated in \Cref{lem:binomial-progress'}.

\begin{lemma}
\label{lem:binomial-progress''}
For all integers \(0\le k\le t\), the progress $\Delta_{t,k}
    \le
    \binom tk
    \left(\frac2{\sqrt N}\right)^k$.
\end{lemma}

\begin{proof}[Proof of \Cref{lem:binomial-progress''}]
It suffices to prove that
$\Delta_{t+1,k+1} \le \Delta_{t,k+1} + \frac2{\sqrt N}\Delta_{t,k}$ for all \(t,k\geq 0\).  Recall that 
\[
    \Delta_{t+1,k+1}
     \leq \Delta_{t,k+1} +
    \norm{\Pi_{\ge k+1}\mathcal{R}\Pi_{=k,\perp}}\cdot \Delta_{t,k}.
\]
\[
\norm{\Pi_{\ge k+1}\mathcal{R}\Pi_{=k,\perp}} = \max_z\norm{\Pi_{\geq k+1} \mathcal{R}\Pi_{=k, \perp} \ket{z}} = \max_z \norm{\Pi_{\geq k+1} \mathcal{R}\ket{z}}
\]
among every basis state $\ket{z}=\ket{t, x,u,w}\ket{r}$ such that $\wt(r)=k$ and $\ket{r_t}=\ket{\perp}$ using the block diagonality of $\mathcal{R}$ similar to the proof of \Cref{lem:progress-recurrence'}.
For every block $t \in [B]$, set up the function $\mathcal{O}_{r,x,u}$ and such that for any $r_t\in [N]$, 
\begin{align*}
    \mathcal{O}_{r,x,u}\ket{r_t}=\omega^{u\cdot \mathbb{1}[x=r_t]} \ket{r_t}=\begin{cases}
        - \ket{r_t},  \quad & u=1 \text{ and } x=r_t\\
        \ket{r_t}, & \text{otherwise }
    \end{cases}
\end{align*}
where $\omega=e^{2\pi i/2}=-1$. Set up the value $\beta_u^{\perp} =\braket{\perpp|S \mathcal{O}_{r,x,u} S |\perpp}$. By using the same argument in the proof of \Cref{clm: beta meaning},  
\[
\norm{\Pi_{\geq k+1} \mathcal{R}\ket{z}}^2 = 1-\abs{\beta_u^\perp}^2.
\]
In addition,
\begin{align*}
    \beta_u^{\perp}  & = \braket{\perpp|S \mathcal{O}_{r,x,u}S|\perpp} = \braket{\tilde{u}| \mathcal{O}_{r,x,u} |\tilde{u}} \\
    & = \frac{1}{N} \sum_{i,j\in [N]}\braket{i|\mathcal{O}_{r,x,u}|j}
    = \frac{1}{N} \sum_{i,j\in [N]}\braket{i|\omega^{u\cdot \mathbb{1}[x=j]}|j}\\
    & =\frac{1}{N} \sum_{j\in [N]}\omega^{u\cdot \mathbb{1}[x=j]} =\frac{1}{N} 
    \left(
    \omega^{u}+ \sum_{j\in [N]: j\neq x} \omega^0
    \right) \\
    & = \frac{1}{N}\left((-1)^{u} + N-1\right) = \begin{cases}
        1-\frac{2}{N} \quad & \text{if } u\neq 0\\
        1 & \text{otherwise} 
    \end{cases}.
\end{align*}
Therefore, for any basis state $\ket{z}=\ket{t, x,u,w}\ket{r}$ with $\wt(r)=k$ and $\ket{r_t}=\ket{\perp}$, 
\begin{align*}
    \norm{\Pi_{\geq k+1} \mathcal{R}\ket{z}}^2 \leq 1-|\beta_u^\perp|^2 \leq 1-\left(1-\frac{2}{N}\right)^2 = \frac{4}{N} - \frac{4}{N^2} \leq \frac{4}{N}.
\end{align*}
It implies that $\norm{\Pi_{\ge k+1}\mathcal{R}\Pi_{=k,\perp}} = \max_z \norm{\Pi_{\geq k+1} \mathcal{R}\ket{z}}\leq \frac{2}{\sqrt{N}}$. Overall, 
\begin{align*}
    \Delta_{t+1,k+1}
    & \leq 1 \cdot \Delta_{t,k+1} +
    \norm{\Pi_{\ge k+1}\mathcal{R}\Pi_{=k,\perp}}\cdot \Delta_{t,k}
    \leq \Delta_{t,k+1} +
    \frac{2}{\sqrt{N}} \Delta_{t,k}
\end{align*}
and it implies that $\Delta_{t,k}
    \le
    \binom tk
    \left(\frac2{\sqrt N}\right)^k$.
\end{proof}

\subsection{From the Recording Progress to the Success Probability}

Let $r=(r_t)_{t\in [B]}$ be an input sampled from $\bigotimes_{t\in [B]} \text{Uniform}([N])$. The task $(B,N, d)$-$\rmif$ is to output the marked item $r_t$ for at least $\lceil dB \rceil$ distinct blocks. 
For every computational basis $\ket{\eta}= \ket{t,x,u,w}_{QPW}$ in the algorithm space $\mathcal{H}_\Alg$, the working register $W$ holds some value $w$ that presents a classical output of the algorithm. Specifically, each basis $\ket{\eta}\in \mathcal{H}_\Alg$ is said to be well formed if it names
\begin{itemize}
    \item a block subset $D_\eta \subseteq [B]$ with $|D_{\eta}|\geq \lceil dB \rceil$ and
    \item a guessing value $(t,g_t)_{t\in D_\eta}$ for each block $t \in D_\eta$.
\end{itemize}
Given an input $r=(r_t)_{t\in [B]}$, we say the basis $\ket{\eta}\in \mathcal{H}_\Alg$ is successful on it if  $g_t=r_t$ for every $t\in D_\eta$.
For each well-formed basis $\ket{\eta}$, we define the projector $Q_\eta$ such that
\[
    Q_\eta
    =
    \left(
    \bigotimes_{t\in D_\eta}
    \proj{g_t}
    \right)
    \otimes 
    \left(
    \bigotimes_{t\in [B]\setminus D_\eta}
    \Id_{\mathcal{H}_t}
    \right).
\]
For each ill-formed basis $\ket{\eta}$ (which violates the output syntax or names less than $\lceil dB \rceil$ blocks), we define $Q_\eta = 0$. The overall success projector is
\begin{equation}\label{eq: pi-success}
        \Pi_{succ}^{RMIF}
    =
    \sum_{\eta}
    \proj{\eta}_{\Alg}\otimes Q_\eta.
\end{equation}
Similarly as in \Cref{lem:main-mbeh helper'} for the BHVF problem, we claim the following lemma and use it to prove \Cref{thm:rmif}.
\begin{lemma}\label{lem:main-mbeh helper}
For every basis $\ket{\eta}\in \mathcal{H}_\Alg$, given any integer $k \leq dB/2$, 
\[
    \norm{Q_\eta\mathcal{T}_B\Lambda_{\le k}}^2
    \le
    e^{-dB/24}
\]
where $\mathcal{T}_B=S^{\otimes B}$ and $\Lambda_{\le k}
    =  \sum_{\substack{L\subseteq[B]:\\|L|\le k}} \Lambda_L$ is defined by \Cref{eq:Projection perp} and \Cref{eq:Lambda-k}.
\end{lemma}
In addition, by analogy with \Cref{lem:technical lemma 1}, we can get the following result by setting $\ket{\perp}$ as the vacuum state.

\begin{lemma}\label{lem:technical lemma 2}
For every $t\in [B]$, let $\mathcal{H}_t= \mathcal{H}_{\rec,\loc}$ be the local space and $\ket{\perpp}$ be the vacuum state.
Let $\{Q_t\}_{t\in D}$ be a given collection of nonzero orthogonal projector with $p_t = \norm{Q_t \ket{\perpp}}^2$. Set up \(\{X_t\}_{t\in D}\) to be mutually independent Bernoulli random variables with $\Prb[X_t=0]=p_t$ and $\Prb[X_t=1]=1-p_t$. Given any integer $k\in \{0,1,\ldots, B\}$, it has
    \[
    \norm{Q_D\Lambda_{\le k}}^2
    = 
    \Prb\left[\sum_{t\in D}X_t\le k\right].
\]
\end{lemma}

\begin{proof}[Proof of \Cref{lem:main-mbeh helper}]
If \(\ket{\eta}\) is ill-formed, then \(Q_\eta=0\) and the claim is correct trivially. 
Assume \(\ket{\eta}\) is well-formed. Let $m=|D_\eta|$.
For each selected block \(t\in D_\eta\), define $\ket{z_t}=S \ket{g_t}$ where $g_t \in [N]$. Then,
\[
    \mathcal{T}_B Q_\eta\mathcal{T}_B
    = 
    S^{\otimes B}Q_\eta S^{\otimes B}
    =
    \left(\bigotimes_{t\in D_\eta}\proj{z_t}
    \right)\otimes \Id_{[B]\setminus D_\eta}.
\]
For each $t\in D_\eta$, set up the projector $Q_t=\proj{z_t}=S\proj{g_t}S$ and $p_t = \norm{Q_t \ket{\perp}}^2$ where
\[
p_t 
= \bra{\perp}Q_t\ket{\perp}
=\abs{\braket{\perpp|z_t}}^2 =\abs{\braket{\perpp|S|g_t}}^2 =\abs{\braket{\tilde{u}|g_t}}^2=\abs{ \frac{1}{\sqrt{N}}}^2=\frac{1}{N}
\]
since $\ket{\tilde{u}}=\frac{1}{\sqrt{N}}\sum_{i\in [N]}\ket{i}$. Set up \(\{X_t\}_{t\in D'_\eta}\) to be mutually independent Bernoulli random variables with $\Prb[X_t=0]=p_t$ and $\Prb[X_t=1]=1-p_t$. 
Given any integer $k\in \{0,1,\ldots, B\}$, by \Cref{lem:technical lemma 2}, 
\[
    \norm{Q_\eta\mathcal{T}_B \Lambda_{\le k}}^2
    =
    \norm{\mathcal{T}_B Q_\eta\mathcal{T}_B \Lambda_{\le k}}^2
    =
    \Prb\left[\sum_{t\in D_\eta}X_t\le k\right].
\]
Since $N\geq 4$, every random variable $X_t$ with $t\in D_\eta$ has $\Prb[X_t=1]=1-\frac{1}{N}\geq 3/4$. Set $\mu =\mathbb{E}[\sum_{t\in D_\eta} X_t]\geq 3m/4$ where $m=|D_\eta|\geq dB$.

For any $k\leq dB/2$, it has $k \leq \frac{1}{2}dB \leq \frac{2}{3}\mu$. Note that $\mu \geq 3dB/4$.
By Chernoff lower-tail bounds, 
\[
\Prb\left[\sum_{t\in D_\eta} X_t \leq k\right]
\leq 
\Prb\left[\sum_{t\in D_\eta} X_t \leq \left(1-\frac{1}{3}\right)\mu \right]\leq 
\exp\left(-\frac{1}{2}\cdot\left(\frac{1}{3}\right)^2\mu \right)
    \le
    e^{-dB/24}.
\]
Therefore $\norm{Q_\eta\mathcal T\Pi_{\le k}}^2
    \le
    e^{-dB/24}$.
\end{proof}
\paragraph{Proof of \Cref{thm:rmif}}
By \Cref{lem:binomial-progress''} and \Cref{lem:main-mbeh helper}, use the same argument as \Cref{thm:bhvf} and we complete the proof.

\section{Deferred Proofs from \Cref{sec:BHVF-lb-proof}}
\subsection{Proof of \Cref{lem:S-matrix}}\label{Proof:S-Matrix}
\begin{proof}
Let
\[
    M
    :=
    \begin{pmatrix}
        0 & \frac1{\sqrt N}\mathbf 1_N^{\mathsf T} \\
        \frac1{\sqrt N}\mathbf 1_N
        &
        \Id_N-\frac1N\mathbf 1_N\mathbf 1_N^{\mathsf T}
    \end{pmatrix}.
\]
It suffices to prove that \(M\) has the same action as \(S\) on the orthogonal decomposition
\[
    \mathcal H_{\loc}
    =
    \spanop\{\ket\perp,\ket{\tilde u}\}
    \oplus
    \spanop\{\ket\perp,\ket{\tilde u}\}^{\perp}.
\]

First, the coordinate vector of \(\ket\perp\) in the ordered basis is \((1\:\:\mathbf{0})^T\)
where the lower block has length \(N\).  Hence
\[
    M\ket\perp
    =
    \begin{pmatrix}
        0\\
        \frac1{\sqrt N}\mathbf 1_N
    \end{pmatrix}
    =
    \frac1{\sqrt N}\sum_{j=0}^{N-1}\ket{e_j}
    =
    \ket{\tilde u}.
\]
Second, the coordinate vector of \(\ket{\tilde u}\) is \((0\:\:\frac1{\sqrt N}\mathbf 1_N)^T\).
Therefore
\[
    M\ket{\tilde u}
    =
    \begin{pmatrix}
        \frac1{\sqrt N}\mathbf 1_N^{\mathsf T}
        \left(\frac1{\sqrt N}\mathbf 1_N\right)\\
        \left(\Id_N-\frac1N\mathbf 1_N\mathbf 1_N^{\mathsf T}\right)
        \left(\frac1{\sqrt N}\mathbf 1_N\right)
    \end{pmatrix}.
\]
The top coordinate is
\[
    \frac1N\mathbf 1_N^{\mathsf T}\mathbf 1_N
    =
    \frac1N\cdot N
    =1.
\]
For the lower block,
we have,
\[
    \left(\Id_N-\frac1N\mathbf 1_N\mathbf 1_N^{\mathsf T}\right)
    \left(\frac1{\sqrt N}\mathbf 1_N\right)
    =
    \frac1{\sqrt N}
    \left(\mathbf 1_N-\frac1N N\mathbf 1_N\right)
    =0.
\]
Consequently
\[
    M\ket{\tilde u}=\begin{pmatrix}1\\0\end{pmatrix}=\ket\perp.
\]

Finally, let \(\ket\phi\perp\spanop\{\ket\perp,\ket{\tilde u}\}\).  Since \(\ket\phi\perp\ket\perp\), its coordinate vector has the form
\[
    \begin{pmatrix}0\\v\end{pmatrix}
    \qquad\text{for some }v\in\mathbb C^N.
\]
The additional condition \(\ket\phi\perp\ket{\tilde u}\) is
\[
    0=\braket{u|\phi}
    =
    \frac1{\sqrt N}\mathbf 1_N^{\mathsf T}v \implies \mathbf 1_N^{\mathsf T}v=0.
\]
Then
\[
    M\ket\phi
    =
    \begin{pmatrix}
        \frac1{\sqrt N}\mathbf 1_N^{\mathsf T}v\\
        \left(\Id_N-\frac1N\mathbf 1_N\mathbf 1_N^{\mathsf T}\right)v
    \end{pmatrix}
    =
    \begin{pmatrix}
        0\\
        v-\frac1N\mathbf 1_N(\mathbf 1_N^{\mathsf T}v)
    \end{pmatrix}
    =
    \begin{pmatrix}0\\v\end{pmatrix}
    =
    \ket\phi.
\]
Thus \(M\) swaps \(\ket\perp\) and \(\ket{\tilde u}\), and fixes \(\spanop\{\ket\perp,\ket{\tilde u}\}^{\perp}\). 
\end{proof}
\subsection{Proof of \Cref{thm:relation-phi-psi}}\label{append:Proof-of-psi-phi}
\begin{proof}
Recall that
\[
\mathcal{T}
=
\mathbb{I}_{QPW}
\otimes
\left((S\otimes S)^{\otimes B}\right),
\qquad
\mathcal{R}
=
\mathcal{T}^{\dagger}\mathcal{O}\mathcal{T}.
\]
Since \(\mathcal{T}\) is unitary, we have

\[
\mathcal{T}\mathcal{R}
=
\mathcal{T}\mathcal{T}^{\dagger}
\mathcal{O}\mathcal{T}
=
\mathcal{O}\mathcal{T}.
\]

Moreover, for every \(i\in\{0,1,\ldots,T\}\), the unitary \(U_i\) acts only on the
algorithm space \(H_{\mathrm{Alg}}\), whereas \(\mathcal{T}\) acts trivially on the
algorithm registers and only on the input registers. Therefore,
\[
\mathcal{T}
\left(U_i\otimes\mathbb{I}_{\mathrm{inp}}\right)
=
\left(U_i\otimes\mathbb{I}_{\mathrm{inp}}\right)
\mathcal{T}.
\]

Recall also that

\[
(S\otimes S)|\perp\perp\rangle
=
|\widetilde{u}\otimes\widetilde{u}\rangle,
\]
and consequently
\[
\mathcal{T}
\left(
|0\rangle_{\mathrm{Alg}}
\left(|\perp\perp\rangle^{\otimes B}\right)_{\mathrm{inp}}
\right)
=
|0\rangle_{\mathrm{Alg}}
\left(
|\widetilde{u}\otimes\widetilde{u}\rangle^{\otimes B}
\right)_{\mathrm{inp}}
=
|0\rangle_{\mathrm{Alg}}
|D\rangle_{\mathrm{inp}}.
\]

Now, from the definition of \(|\phi_T\rangle\), we have
\begin{align*}
\mathcal{T}|\phi_T\rangle
&=
\mathcal{T}
\left(U_T\otimes\mathbb{I}_{\mathrm{inp}}\right)
\mathcal{R}
\left(U_{T-1}\otimes\mathbb{I}_{\mathrm{inp}}\right)
\cdots
\mathcal{R}
\left(U_0\otimes\mathbb{I}_{\mathrm{inp}}\right)
|0\rangle_{\mathrm{Alg}}
\left(|\perp\perp\rangle^{\otimes B}\right)_{\mathrm{inp}}
\\
&=
\left(U_T\otimes\mathbb{I}_{\mathrm{inp}}\right)
\mathcal{T}\mathcal{R}
\left(U_{T-1}\otimes\mathbb{I}_{\mathrm{inp}}\right)
\cdots
\mathcal{R}
\left(U_0\otimes\mathbb{I}_{\mathrm{inp}}\right)
|0\rangle_{\mathrm{Alg}}
\left(|\perp\perp\rangle^{\otimes B}\right)_{\mathrm{inp}}
\\
&=
\left(U_T\otimes\mathbb{I}_{\mathrm{inp}}\right)
\mathcal{O}\mathcal{T}
\left(U_{T-1}\otimes\mathbb{I}_{\mathrm{inp}}\right)
\cdots
\mathcal{R}
\left(U_0\otimes\mathbb{I}_{\mathrm{inp}}\right)
|0\rangle_{\mathrm{Alg}}
\left(|\perp\perp\rangle^{\otimes B}\right)_{\mathrm{inp}}
\\
&=
\left(U_T\otimes\mathbb{I}_{\mathrm{inp}}\right)
\mathcal{O}
\left(U_{T-1}\otimes\mathbb{I}_{\mathrm{inp}}\right)
\mathcal{T}\mathcal{R}
\cdots
\mathcal{R}
\left(U_0\otimes\mathbb{I}_{\mathrm{inp}}\right)
|0\rangle_{\mathrm{Alg}}
\left(|\perp\perp\rangle^{\otimes B}\right)_{\mathrm{inp}}
\\
&=
\cdots
\\
&=
\left(U_T\otimes\mathbb{I}_{\mathrm{inp}}\right)
\mathcal{O}
\left(U_{T-1}\otimes\mathbb{I}_{\mathrm{inp}}\right)
\cdots
\mathcal{O}
\left(U_0\otimes\mathbb{I}_{\mathrm{inp}}\right)
\mathcal{T}
|0\rangle_{\mathrm{Alg}}
\left(|\perp\perp\rangle^{\otimes B}\right)_{\mathrm{inp}}
\\
&=
\left(U_T\otimes\mathbb{I}_{\mathrm{inp}}\right)
\mathcal{O}
\left(U_{T-1}\otimes\mathbb{I}_{\mathrm{inp}}\right)
\cdots
\mathcal{O}
\left(U_0\otimes\mathbb{I}_{\mathrm{inp}}\right)
|0\rangle_{\mathrm{Alg}}
|D\rangle_{\mathrm{inp}}
\\
&=
|\psi_T\rangle.
\end{align*}

\end{proof}

\end{document}